\documentclass[a4paper, UKenglish, cleveref, autoref, thm-restate]{lipics-v2021}

\usepackage{mathtools}
\usepackage[T1]{fontenc}
\usepackage{array}
\usepackage{amsthm}
\usepackage{amsmath}
\usepackage{amstext}
\usepackage{amsfonts}
\usepackage{amssymb}
\usepackage{graphicx}
\usepackage{setspace}
\usepackage{esint}
\usepackage{subcaption}
\usepackage{dsfont}
\usepackage{placeins}
\usepackage{fourier}
\usepackage{cleveref}
\usepackage{harpoon}

\usepackage[colorinlistoftodos]{todonotes}

\newcommand{\cO}{{\mathcal{O}}}
\newcommand{\bA}{{\mathbb A}}
\newcommand{\bB}{{\mathbb B}}
\newcommand{\bZ}{{\mathbb Z}}
\newcommand{\tcorr}{\text{corr}}
\newcommand{\flip}{\mathrm{flip}}

\newcommand{\elide}[1]{}

\makeatletter

\theoremstyle{plain}
\newtheorem{thm}{\protect\theoremname}
\numberwithin{thm}{section}

\theoremstyle{plain}
\newtheorem{lem}{\protect\lemmaname}
\numberwithin{lem}{section}

 \theoremstyle{plain}
\newtheorem{rem}{\protect\remarkname}
\numberwithin{rem}{section}

\theoremstyle{plain}
\newtheorem{clm}[thm]{\protect\claimname}

\theoremstyle{plain}
\newtheorem{corr}[thm]{\protect\corrname}

\theoremstyle{plain}
\newtheorem{obs}{\protect\obsname}
\numberwithin{obs}{section}

\theoremstyle{definition}
\newtheorem{defn}{\protect\definitionname}
\numberwithin{defn}{section}

\AtBeginDocument{%
  \expandafter\renewcommand\expandafter\subsubsection\expandafter
    {\expandafter\@fb@secFB\subsubsection}%
  \newcommand\@fb@secFB{\FloatBarrier
    \gdef\@fb@afterHHook{\@fb@topbarrier \gdef\@fb@afterHHook{}}}%
  \g@addto@macro\@afterheading{\@fb@afterHHook}%
  \gdef\@fb@afterHHook{}%
}

\makeatother
\nolinenumbers

\usepackage{babel}

\providecommand{\definitionname}{Definition}

\providecommand{\obsname}{Observation}
\providecommand{\theoremname}{Theorem}
\providecommand{\lemmaname}{Lemma}
\providecommand{\corrname}{Corollary}
\providecommand{\remarkname}{Remark}

\providecommand{\claimname}{Claim}

\title{Minimum-length coordinated motions for two 
convex centrally-symmetric robots}

\author{David Kirkpatrick\footnote{Corresponding author.}}{University of British Columbia}{kirk@cs.ubc.ca}{https://orcid.org/0000-0002-3276-2734}{Supported by  Discovery Grant from the Natural Sciences and Engineering Research Council of Canada}

\author{Paul Liu}{Stanford University \and \url{cs.stanford.edu/people/paulliu}}{paul.liu@stanford.edu}{https://orcid.org/0000-0002-9386-6609}{}

\authorrunning{D. Kirkpatrick and P. Liu} 
\Copyright{David Kirkpatrick and Paul Liu} 

\relatedversion{}
\relatedversiondetails{Full Version}{link-to-arxiv}

\acknowledgements{We are indebted to Will Evans and Dan Halperin for encouragement and helpful discussions. }

\date{}


\begin{document}
\maketitle

\begin{abstract}
We study the problem of determining 
coordinated motions, of minimum total length, for two arbitrary convex centrally-symmetric (CCS) robots in an otherwise obstacle-free plane. Using the total path length traced by the two robot centres as a measure of distance, we give an exact characterization of a (not necessarily unique) shortest collision-avoiding motion for all initial and goal configurations of the robots. 
The individual paths are composed of at most six convex pieces, and their total length can be expressed as a simple integral with a closed form solution depending only on the initial and goal configuration of the robots. 
The path pieces are either straight segments or segments of the boundary of the Minkowski sum of the two robots (circular arcs, in the special case of disc robots).
Furthermore, the paths can be parameterized in such a way that (i) only one robot is moving at any given time (decoupled motion), or (ii) the 
orientation of the robot configuration 
changes monotonically.
\end{abstract}

\section{Introduction}

Given
a collection of robots\footnote{Open sets in $\mathbb{R}^2$.} in the plane, each with specified initial and goal configurations, we are interested in the problem of identifying efficient collision-free coordinated translational motions 
taking all of the robots 
from their initial to their goal configuration. 
In general, the cost of such a coordinated motion is defined to be the length sum of the paths traced by the centre (or arbitrary fixed point) of 
each of the robots.
This problem has a rich history dating back to the early 1980s, much of it devoted to the special cases of disc and square robots.
See \cite{Antonyshyn2023} for a recent survey.
We begin by outlining some of the existing work, the bulk of which is focused on the feasibility, rather than optimality, of coordinated motions. 

Schwartz and Sharir \cite{sharir0} were the first to study coordinated motion planning for $k$ discs among polygonal obstacles with $n$ total edges. For $k=2$, they developed an $\cO(n^3)$ algorithm (later improved to $\cO(n^2)$)~\cite{sharir2, yap} to determine if a collision-free coordinated motion connecting two specified configurations is possible. When the number of robots is unbounded, Spirakis and Yap \cite{spirakis} showed that determining feasibility is strongly NP-hard for disc robots, although the proof relies on the robots having different radii. For the analogous problem with rectangular robots, determining feasibility is PSPACE-hard, as shown by Hopcroft et al. \cite{hopcroft0} and Hopcroft and Wilfong \cite{hopcroft1}. This result was later generalized by Hearn and Demaine \cite{hearn} for rectangular robots of size $1\times 2$ and $2 \times 1$. In a recent work by Fekete et al.~\cite{fekete24}, the setting of axis-parallel bounded domains was explored, where many robots of size $1 \times 1$ can move on a lattice grid. In this setting, Fekete et al. provided a variety of results, including a characterization of domains where feasibility is possible, as well as algorithms for motion planning.

On the practical side, heuristic and sampling based algorithms have been employed to solve coordinated motion planning problem for up to hundreds of robots~\cite{gildardo,standley,wagner}. These algorithms typically use standard search strategies such as $A^*$ coupled with domain specific heuristics (see~\cite{planning-text} and the references contained therein). While efficient in practice, these algorithms are typically numerical or iterative in nature, with no precise performance bounds. 
A variety of alternative cost measures for our problem has also been considered, such as the minimum elapsed time coordinated motion under velocity constraints~\cite{ChenIerardi, ladder, turpin} as well as the coordinated motion minimizing the total number of continuous movements~\cite{abellanas, bereg, dumitrescu}. 
The motion planning problem for square robots in a rectilinear setting has also been the focus of the recent CG:SHOP 2021 Challenge; a variety of practical methods for addressing this problem have been proposed~\cite{crombez22, fekete22, liu22, yang22}.

A variant of the problem considers robots that are homogeneous and unlabelled. In this case, any robot is allowed to move to any target location, so long as each target position is covered by exactly one robot. For two discs, the unlabelled case is trivial as one can apply our labelled algorithm twice. However, when the number of discs is unbounded, Solovey and Halperin \cite{solovey2} show that the unlabelled problem is PSPACE-hard, even in the case of unit squares with polygonal obstacles. 
Surprisingly, when the robots are located within a simple polygon with no obstacles, a polynomial time for checking feasibility exists \cite{Adler2015}. 
As in the labelled case, a variety of cost measures has been explored for the unlabelled case. Solovey et al.~\cite{solovey1} gives an $\tilde{\cO}(k^4+k^2n^2)$ algorithm that minimizes the length sum of paths traced by the centres of $k$ discs, with additive error $4k$. 
In work by Turpin et al. \cite{turpin}, an optimal solution is found in polynomial time when the cost function is the maximum path length traversed by any single robot. However, their algorithm requires that the working space is obstacle free and the initial locations of the robots are far enough apart. 
Recent work of Agarwal et al.~\cite{Agarwal+SoDA24} considers algorithms that approximate the minimum total length of collision-free coordinated motions of two axis-aligned square robots. 
They provide 
the first polynomial-time $(1 + \varepsilon)$-approximation algorithm
for an optimal motion-planning problem involving two independent robots moving in a polygonal environment.

In this paper, we focus on the minimum-length coordinated motion of two convex centrally-symmetric (CCS) robots in an otherwise obstacle-free plane. 
The special case of disc robots was treated previously by Kirkpatrick and Liu~\cite{kl2016, liu2017}). Using essentially the same techniques, 
Esteban et al.~\cite{esteban23} addressed the special case of congruent (i.e. similarly aligned) square robots.   
Our 
comprehensive description 
of coordinated motions for CCS robots, between arbitrary initial and goal placements, 
subsumes and simplifies the results in these earlier papers, and in so doing highlights both the generality and limitations of the approach initiated in~\cite{kl2016, liu2017}. 
As before, the analysis underlying our characterization makes use of the Cauchy surface area formula, which was introduced by Icking et al.~\cite{icking}, in describing optimal motions for rods (directed line segments) in the plane,
where distance is measured by the length sum of the paths traced by the two endpoints of the segment.

We note that the rod motion problem itself 
has a rich history, and was first posed by Ulam \cite{ulam} and subsequently solved by Gurevich \cite{gurevich}. Other approaches to that of Icking et al. are quite different, and use control theory to obtain differential equations that characterize the optimal motion~\cite{gurevich, verriest}. Of course, the problem of moving a directed line segment of length $s$ corresponds exactly to the coordinated motion of two discs with radius sum $s$ constrained to remain in contact throughout the motion. Hence the coordinated motion of two discs with radius sum $s$ can also be seen as the problem of moving an ``extensible" line segment that can extend freely but has minimum length $s$. As such, our results also generalize those of Icking et al.  Although we use some of the same tools introduced by Icking et al., our generalization is non-trivial; the doubling argument that lies at the heart of the proof of Icking et al.~depends in an essential way on the assumption that the rod length is fixed throughout the motion. 

\subsection{Our contributions}
In contrast to previous work focused on disc or square shaped robots,
we consider arbitrary CCS shapes. 
In addition the two robots do not have to be congruent, for example both circles or both squares. The work presented here is both a simplification and generalization of the results presented in \cite{esteban23, kl2016}. Similar in spirit to previous results, 
our characterization result identifies cases where the coordinated motion can be decomposed into a few simple segments and proves optimality by an application of Cauchy's surface area formula. 
The main observation that 
underpins our generalization
is that the motion decomposition has a simple expression in terms of the Minkowski sum of the two robots, rather than appealing to specific geometric properties of the circle or the square.

For the case of two arbitrary CCS robots, 
our approach first characterizes all initial and goal robot configurations that admit straight-line 
coordinated motions. 
For all other initial and goal configurations, the coordinated motion from initial to goal configuration involves either a net clockwise or counter-clockwise turn in the relative position of the robots. 
In this case, our results describe either (i) a single globally optimal coordinated motion, or (ii) two coordinated motions, of which one is optimal among all net clockwise coordinated motions and the other is optimal among all net counter-clockwise coordinated motions. 
For all of the optimal coordinated motions that we describe,
the trace of both robots has the same simple form: a constant number of straight and convex segments, where the convex segments are formed from a contiguous piece of the boundary of the Minkowski sum of 
the two robots. 
Our constructions of shortest 
coordinated motions can be summarized by the following theorem:

\begin{thm}
\label{thm:mainthm}
Let $\bA$ and $\bB$ be two arbitrary convex centrally-symmetric (CCS) robots. For any initial and goal configurations of $\bA$ and $\bB$, there is a 
minimum-length collision-free coordinated motion, composed of at most six pieces, taking $\bA$ and $\bB$ from the initial to the goal configuration. 
The pieces are either straight segments or contiguous portions of the boundary of the Minkowski sum of $\bA$ and $\bB$.
\end{thm}

The length of the shortest coordinated motion can be expressed as a simple integral depending only on the initial and goal configurations. 
Moreover, all coordinated motion traces that we describe can be realized by two different kinds of coordinated motion: \emph{coupled} or \emph{decoupled}. 
In the decoupled case, the coordinated motion consists of three phases, with only one of the robots moving in each phase. 
In the coupled case, the angle 
describing the relative orientation of the robots 
changes monotonically.
Furthermore, the motions can be parameterized in a way that the  two robots are in contact for a contiguous interval of time 
(i.e., once the two robots move out of contact, they are never in contact again.) 

The rest of the paper is organized as follows. 
In Section~\ref{sec:backgrd} we present some basic definitions and describe some of the essential tools used in our characterization and correctness proofs. In Section~\ref{sec:overview} we summarize the general structure of our proofs. Section~\ref{sec:mainprf} provides intuition for the proofs by providing illustrative optimal motions for various cases. These cases are designed to communicate the heart of the proof, while avoiding the detailed case analysis. The full analysis is given in Section~\ref{sec:modified}. 

\begin{figure}[ht]
\centering
\includegraphics[width=0.9\textwidth]{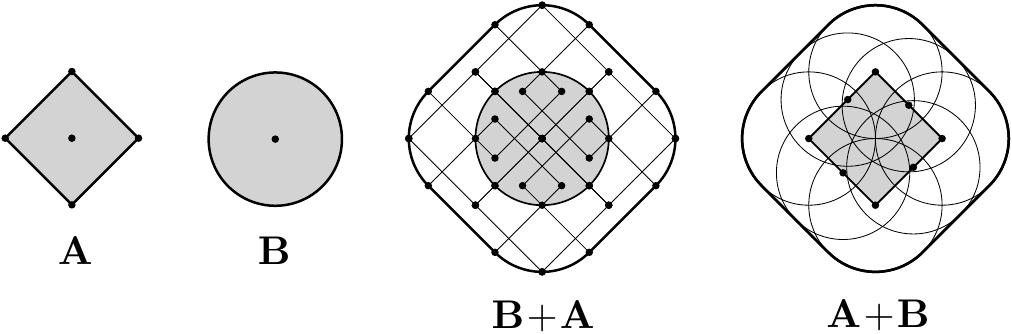}
\caption{Two CCS objects and their Minkowski sums. 
}. \label{fig:Minkowski0}
\end{figure}

\begin{figure}[ht]
\centering
\includegraphics[width=0.9\textwidth]{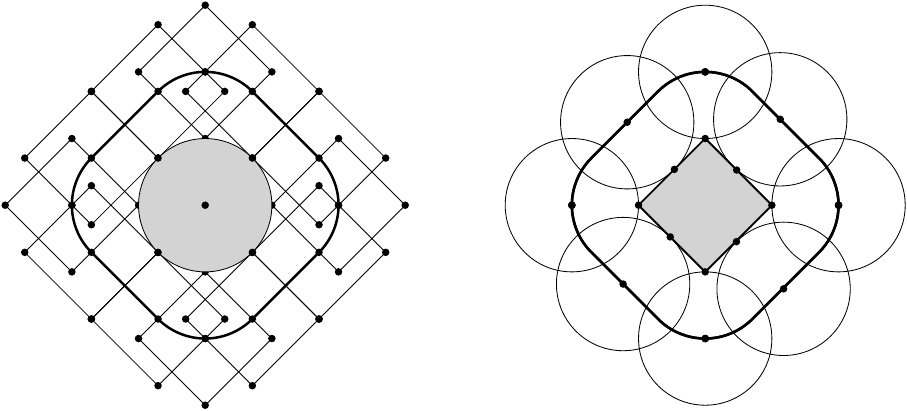}
\caption{Minkowski sums as the locus of intersecting placements of CCS objects. 
}. \label{fig:Minkowski1}
\end{figure}

\section{Background}
\label{sec:backgrd}

To describe the coordinated motion of a pair of CCS robots between their initial and goal configurations, we first make precise several terms that have intuitive meaning. 

\begin{defn}
A region $\mathbb S \subset \Re^2$ is called \emph{centrally-symmetric} if 
$(x,y) \in \mathbb S  \implies (-x, -y) \in \mathbb S$.
\end{defn}

\begin{defn}
The \emph{Minkowski sum} of two regions $R_1, R_2 \subset \Re^2$ is the region $R_1\!+\!R_2 := \{p_1 + p_2 \,\mid \, p_1 \in R_1, p_2 \in R_2\}$ (Figure~\ref{fig:Minkowski0} and \ref{fig:Minkowski1}).
Similarly, the \emph{Minkowski difference} of $R_1, R_2 \subset \Re^2$ is the region $R_1\!-\!R_2 := \{p_1 - p_2 \,\mid \, p_1 \in R_1, p_2 \in R_2\}$. Note that if $R_1$ and $R_2$ are centrally-symmetric, then 
$R_1\!+\!R_2 = R_1\!-\!R_2$.
\end{defn}

\begin{defn}
The \emph{reach} of a set of points $\mathbb S$ in direction $\theta$, is given by 
\[r_{\mathbb S}(\theta) := \sup \{ x \cos \theta + y \sin \theta : (x,y) \in \mathbb S \}.\] 
For an angle $\theta$, the set of points that realize the supremum are called \emph{support points}, and the line oriented at angle $\frac{\pi}{2}+\theta$ passing through the support points is called the \emph{support line} (see Figure~\ref{fig:reach}). When the set $\mathbb S$ consists of a single point $S$, we write $r_S$ instead of $r_{\mathbb S}$ for ease of notation.
\end{defn}

\begin{obs}\label{obs:oppositereach}
Note that $r_{\mathbb S}(\pi +\theta) = r_{-\mathbb S}(\theta)$.
Furthermore, $r_{\mathbb S}(\theta) = r_{\widearc{\mathbb S}}(\theta)$, where 
$\widearc{\mathbb S}$ denotes the boundary of the convex hull of $\mathbb S$.
\end{obs} 

\begin{figure}[ht]
\centering
\includegraphics[ scale=0.75]{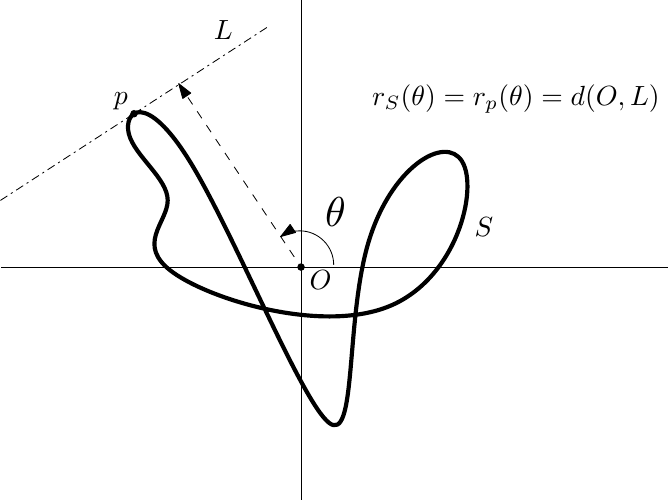}
\caption{Support line $L$ and support point $p$ defining the reach of $S$ in direction $\theta$. 
}. \label{fig:reach}
\end{figure}

\begin{figure}[h!]
\centering
\includegraphics[ scale=0.75]{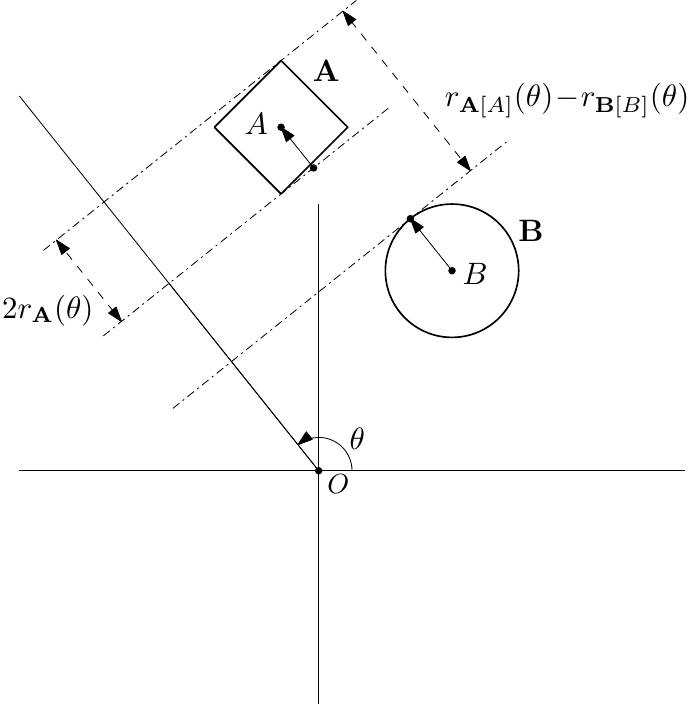}
\caption{The separation of $\bA[A]$ from $\bB[B]$ in direction $\theta$. 
} \label{fig:separation}
\end{figure}

\begin{defn}
A generic \emph{CCS robot}, hereafter simply \emph{robot}, is a convex and centrally-symmetric region in the plane.
The \emph{placement} of robot $\bZ$ at position $Z$ in $\Re^2$, denoted $\bZ[Z]$, is formed by translating $\bZ$ so that its centre coincides with $Z$.  
Accordingly, $r_{\bZ[Z]}(\theta)= r_Z(\theta) + r_{\bZ}(\theta)$. 
We refer to $\bZ[O]$, the placement at the origin $O$, as the \emph{standard} placement of $\bZ$.
\end{defn}

\begin{defn}
An \emph{(instantaneous) configuration} of a pair of robots $(\bA, \bB)$ is just a pair of placements
$(\bA[A], \bB[B])$. 
In configuration
$(\bA[A], \bB[B])$, the \emph{separation} of $\bA[A]$ from $\bB[B]$
in direction $\theta$ is
given by 
$r_{\bA[A]}(\theta) 
- r_{\bB[B]}(\theta)
-2 r_{\bA}(\theta)$;
cf. Figure~\ref{fig:separation}.
\end{defn}

\begin{defn}
A configuration $(\bA[A], \bB[B])$ is \emph{viable} if 
$ \bA[A] \cap \bB[B] = \emptyset$;
equivalently if the separation of
$\bA[A]$ from $\bB[B]$ is non-negative for some direction $\theta$. 
For any given viable configuration $(\bA[A], \bB[B])$, 
we refer to the direction $\theta$ that maximizes the separation  
of $\bA[A]$ from $\bB[B]$ 
as its \emph{orientation}.
\end{defn}

\begin{obs}\label{obs:Minkowski}
 The Minkowski sum of two CCS robots $\bA$ and $\bB$,
 denoted $\bA\!+\!\bB$, is another CCS robot, whose reach $r_{\bA\!+\!\bB}(\theta)$ satisfies $r_{\bA\!+\!\bB}(\theta) = r_{\bA}(\theta) + r_{\bB}(\theta)$. 
 By symmetry, $\bA[A]$ intersects $\bB[B]$ if and only if point $A$ lies in $(\bA\!+\!\bB)[B]$ 
 (equivalently, point $A-B$ lies in $(\bA\!+\!\bB)[O]$);
 see Figure~\ref{fig:Minkowski1}.
 
 If point $\bA[A]$ does not intersect $\bB[B]$,
 then 
 (i) the separation of $\bA[A]$ from $\bB[B]$ is just the distance from $A$ to the boundary of $(\bA\!+\!\bB)[B]$, and
 (ii) the orientation of the
 configuration $(\bA[A], \bB[B])$ is just the angle, with respect to the $x$-axis, formed by the  outward normal to $(\bA\!+\!\bB)[B]$ through point $A$ 
 (equivalently, the outward normal to $(\bA\!+\!\bB)[O]$ through point $A-B$).
 Note that if $\bA$ and $\bB$ are discs, then the orientation of the
 configuration $(\bA[A], \bB[B])$ is the angle 
 $\phi_{\overrightarrow{BA}}$ formed by the vector 
 $\overrightarrow{BA}$ with the $x$-axis,
but this is not true in general. 
(For example, in Figure~\ref{fig:Minkowski1} the orientation of contact configurations changes only when the contact on $\bA$ is a corner point.)
 \end{obs}

A pair of robots can move from configuration to configuration through a coordinated motion, which we can now define:

\begin{defn}
A (translational)
\emph{motion} 
$\xi_{\bZ}$ of a robot ${\bZ}$ from placement $\bZ[Z_0]$ (i.e. position $Z_0$) to
placement $\bZ[Z_1]$ (i.e. position $Z_1$) 
is a continuous, rectifiable curve of the form $\xi_{\bZ} : [0,1] \rightarrow \Re^2$, where $\xi_{\bZ}(0) = Z_0$, and $\xi_{\bZ}(1) = Z_1$.
The set of points 
$tr(\{\xi_{\bZ}(t) \;|\; 0\le t\le 1 \})$
is called the \emph{trace} of robot $\bZ$ under motion $\xi_{\bZ}$. 
\end{defn}

\begin{defn}
The \emph{length} of motion $\xi_{\bZ}$, denoted $\ell(\xi_{\bZ})$, is simply the Euclidean arc-length of its trace, that is,
\[
	\ell(\xi_{\bZ}) = \sup_{T} \sum_{i=1}^k ||\xi_{\bZ}(t_{i-1})- \xi_{\bZ}(t_i)||
\]
where the supremum is taken over all subdivisions $T=\{t_0,t_1,\ldots,t_k\}$ of $[0,1]$ where $0 = t_0 < t_1 < \cdots < t_k = 1$. 
\end{defn}

\begin{obs}\label{obs:convex}
Denote by $\widearc{\xi_{\bZ}}$ the closed curve defining the boundary of the convex hull of the trace of $\xi_{\bZ}$. 
Since
the trace of $\xi_{\bZ}$, together with the segment $\overline{Z_0Z_1}$, forms a closed curve whose convex hull has boundary  $\widearc{\xi_{\bZ}}$, it follows from convexity that:
\begin{equation}
\label{eq:inequal}
\ell(\xi_{\bZ})\geq \ell(\widearc{\xi_{\bZ}}) - |\overline{Z_0Z_1}|
\end{equation}
\end{obs}

\begin{defn}
A motion $\xi_{\bZ}$ 
is said to be \emph{convex} if, 
equation~(\ref{eq:inequal}) is an equality; i.e. 
$\ell(\widearc{\xi_{\bZ}}) = \ell(\xi_{\bZ}) + |\overline{Z_0Z_1}|$.
\end{defn}

\begin{obs}
It follows from Observation~\ref{obs:Minkowski} that a motion of robot $\bA$ from placement $\bA[A]$ to
placement $\bA[A']$ avoiding robot $\bB$ at placement $\bB[B]$ can be viewed as the motion of a point (the centre of $\bA$), from $A$ to $A'$ avoiding $(\bA\!+\!\bB)[B]$.
\end{obs}

\begin{defn}
A \emph{coordinated motion} (hereafter \emph{co-motion}) $m$ of a robot pair $({\bA},{\bB})$ from an initial configuration  
$(\bA[A_0], \bB[B_0])$ to a goal configuration 
$(\bA[A_1], \bB[B_1])$
is a pair $(\xi_{\bA}, \xi_{\bB})$, where $\xi_{\bA}$ (resp. $\xi_{\bB}$) is a motion of $\bA$ (resp. $\bB$) from placement $\bA[A_0]$ to placement $\bA[A_1]$ (resp. placement $\bB[B_0]$ to placement $\bB[B_1]$). 
The co-motion $m$ is said to be 
\emph{collision-free}
if 
the configuration 
$(\bA[\xi_{\bA}(t)], \bB[\xi_{\bB}(t)])$ is
viable, for all $t \in [0,1]$.
Co-motion $m=(\xi_{\bA}, \xi_{\bB})$ is said to be \emph{convex} if both $\xi_{\bA}$ and $\xi_{\bB}$ are convex. Its length, denoted $\ell(m)$, 
is the sum of the lengths of its associated motions, i.e. $\ell(m) = \ell(\xi_{\bA}) + \ell(\xi_{\bB})$. 
\end{defn}

\begin{rem}
   Co-motions $m'=(\xi'_{\bA}, \xi'_{\bB})$ and $m=(\xi_{\bA}, \xi_{\bB})$ are said to be \emph{trace-equivalent} if the trace of $\xi'_{\bA}$ is the same as the trace of $\xi_{\bA}$ and the trace of
   $\xi'_{\bB}$ is the same as the trace of $\xi_{\bB}$. 
   Note that, for trace-equivalent co-motions $m$ and $m'$, $m'$ is convex if and only if $m$ is convex, and $\ell(m') = \ell(m)$. However, if $m$ is collision-free it does not necessarily follow that $m'$ is collision-free.
\end{rem}

\begin{figure}[h!]
\centering
\includegraphics[width=0.9\textwidth]
{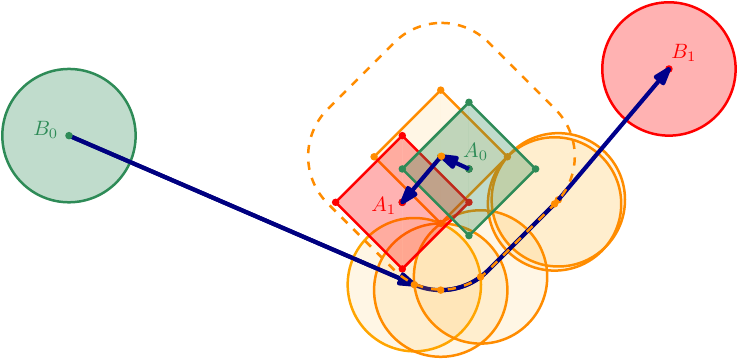}
\caption{The trace of a co-motion from initial configuration (green) to target configuration (red), with intermediate configurations (yellow)}. \label{fig:introexample}
\end{figure}

Figure~\ref{fig:introexample} illustrates a co-motion (blue) of two robots from their initial configuration (green) to a goal configuration (red). 
The diamond robot first moves (straight) to its intermediate placement (orange). Then the disc robot moves around the diamond robot to its goal placement, passing through intermediate placements with boundary contact. 
Finally, the diamond robot moves (straight) to its goal placement.

\begin{defn}
The \emph{(collision-free) distance} between two configurations $\bA\bB_0=(\bA[A_0], \bB[B_0])$ and $\bA\bB_1=(\bA[A_1], \bB[B_1])$,
denoted $d(\bA\bB_0,\bA\bB_1)$, is the minimum possible length over all collision-free co-motions $m$ from $\bA\bB_0$ to $\bA\bB_1$. 
We refer to any collision-free co-motion $m$ between $\bA\bB_0$ and $\bA\bB_1$ satisfying $\ell(m) = d(\bA\bB_0,\bA\bB_1)$ as an \emph{optimal} co-motion between $\bA\bB_0$ and $\bA\bB_1$. 
\end{defn}

The fact that $d$ is a metric on the set of configurations is easy to check. Nevertheless, one may be concerned about the existence of  a shortest co-motion under this notion of distance.  The fact that a shortest co-motion exists is a consequence of the Hopf-Rinow theorem, for which details can be found in \cite{gromov}.

\section{The general approach}
\label{sec:overview}

For the remainder of the paper, we will assume a fixed CCS robot pair 
$(\bA, \bB)$. 
We denote by $\bA\bB_0$ the initial configuration
$(\bA[A_0], \bB[B_0])$, and by $\bA\bB_1$ the goal configuration
$(\bA[A_1], \bB[B_1])$.
We denote by $[\theta_0, \theta_1]$ the range of angles counter-clockwise between the orientation $\theta_0$ of $\bA\bB_0$ and the orientation $\theta_1$ of $\bA\bB_1$.
Finally, we use $I_m$ to denote the range of orientations realized by the set of all configurations associated with a given co-motion $m$.
The following observation allows us to categorize co-motions into two types: 
\emph{net counter-clockwise}
co-motions satisfy $[\theta_0,\theta_1]\subseteq I_m$, and
\emph{net clockwise} co-motions satisfy
$S^1-[\theta_0,\theta_1]\subseteq I_m$, where $S^1 = [0,2\pi]$.
Co-motions satisfying $[\theta_0,\theta_1] = I_m$ (resp., 
$S^1-[\theta_0,\theta_1]= I_m$) are said to be \emph{strictly counter-clockwise}
(resp., \emph{strictly clockwise}) co-motions.

\begin{obs}
\label{monotone} Let $m = (\xi_{\bA}, \xi_{\bB})$ be any co-motion 
taking $(\bA, \bB)$ from 
initial configuration $\bA\bB_0$ and goal configuration 
$\bA\bB_1$.
By the continuity of $\xi_{\bA}$ and $\xi_{\bB}$, 
it follows that $[\theta_0,\theta_1]\subseteq I_m$ or $S^1-[\theta_0,\theta_1]\subseteq I_m$.
That is, any co-motion is either net clockwise or net counter-clockwise (or both).
\end{obs}

\begin{obs}\label{obs:flip}
If $S$ is any planar point set, denote by $\flip(S)$, the reflection of $S$ across the $x$-axis, i.e. the set $\{ (x, -y) \;|\; (x, y) \in S \}$.
Note that if $m$ is a net clockwise co-motion from configuration $\bA\bB_0$ to configuration $\bA\bB_0$, then $\flip(m)$ is a net counter-clockwise co-motion of the robot pair $(\flip(\bA), \flip(\bB))$ from configuration 
$\flip(\bA\bB_0)$ to $\flip(\bA\bB_1)$.
Thus, by the preceding Observation, in order to characterize optimal co-motions it will suffice to focus on net counter-clockwise co-motions.
In fact, as we will see, every minimum-length collision-free co-motion has a trace-equivalent collision-free co-motion whose associated configurations have orientations that change monotonically. So we can focus attention on strictly counter-clockwise co-motions.
\end{obs}

It follows directly from Observation~\ref{obs:convex} that
to determine the length of a convex co-motion $(\xi_{\bA}, \xi_{\bB})$ it suffices to be able to determine the length of the pair of closed convex curves 
$(\widearc{\xi_{\bA}}, \widearc{\xi_{\bB}})$.  
For this we follow an approach, first described by Icking et al.~\cite{icking}, that allows us to express the the length of two convex curves in terms of the reach of those curves in all directions. 
This expression is a direct application of Cauchy's surface area formula:

\begin{lem}
\label{lem:cauchyorig}
(Cauchy's surface area formula \cite[Section 5.3]{egg}) Let $C$ be a closed convex curve in the plane and $r_C(\theta)$ be the reach of $C$ in direction $\theta$. Then 
\begin{equation}
\ell(C) = \int_{S^1} r_C(\theta)d\theta.
\end{equation}
\end{lem}

As noted in \cite{icking}, it follows directly from Lemma~\ref{lem:cauchyorig} that we can find a similarly simple expression for the sum of the lengths of two convex curves in the plane:
\begin{corr}
\label{cor:cauchyorig}
Let $C_{1}$ and $C_{2}$ be closed convex curves in the plane. 
Then the
sum of their lengths can be expressed as: \label{thm:cauchy} 
\begin{equation}
\ell(C_{1})+\ell(C_{2}) = \int_{S^1}\left(r_{C_1}(\theta)+r_{C_2}(\pi+\theta)\right)d\theta. \label{eq:cauchy}
\end{equation}
\end{corr}

\begin{obs}\label{obs:Cauchy2}
It follows from Observation~\ref{obs:oppositereach} that (i) $r_{C_1}(\theta)+r_{C_2}(\pi+\theta) = 
r_{C_1\!-\!C_2}(\theta)$, and
(ii) $r_{C_1\!-\!C_2}(\theta) = r_{\widearc{C_1\!-\!C_2}}(\theta)$, where $\widearc{C_1\!-\!C_2}$ denotes the boundary of the convex hull of $C_1\!-\!C_2$. Thus, by another application of 
Lemma~\ref{lem:cauchyorig},
\begin{equation}
\ell(C_1) + \ell(C_2) = \ell(\widearc{C_1\!-\!C_2}). \label{eq:cauchy2}
\end{equation}
\end{obs}

The following observations show how basic properties of the co-motion $(\xi_{\bA}, \xi_{\bB})$ are reflected in the reach functions associated with the pair of convex curves $(\widearc{\xi_{\bA}}, \widearc{\xi_{\bB}})$.

\begin{obs}
\label{obs:point-wise}
If $m = (\xi_{\bA}, \xi_{\bB})$ is an arbitrary collision-free co-motion then
$r_{\widearc{\xi_{\bA}}}(\theta)
+r_{\widearc{\xi_{\bB}}}(\pi+\theta)$ 
describes the maximum separation, in direction $\theta$, realized by a point on $\widearc{\xi_{\bA}}$ and a point on   $\widearc{\xi_{\bB}}$.
If $\theta$ lies in $I_m$,
the range of orientations realized by the set of all configurations associated with $m$,
then for some $t \in [0,1]$, 
the configuration 
$(\bA[\xi_{\bA}(t)], \bB[\xi_{\bB}(t)]) $ has orientation $\theta$, and hence by the viability property and central-symmetry, 
$r_{\widearc{\xi_{\bA}}}(\theta)
+r_{\widearc{\xi_{\bB}}}(\pi+\theta) 
\ge r_{\xi_{\bA}(t)}(\theta)
+r_{\xi_{\bB}(t)}(\pi+\theta) 
\ge r_{\bA}(\theta) + r_{\bB}(\pi +\theta)
= r_{\bA\!-\!\bB}(\theta)
= r_{\bA\!+\!\bB}(\theta).$
\end{obs}

\begin{obs} Suppose that $m$ is any co-motion from 
configuration $\bA\bB_0$ to configuration $\bA\bB_1$.
Since the reach of $\widearc{\xi_{\bA}}$ is at least the reach of the segment $\overline{A_0A_1}$ and 
the reach of $\widearc{\xi_{\bB}}$ is at least the reach of the segment $\overline{B_0B_1}$
it follows that
$r_{\widearc{\xi_{\bA}}}(\theta)
+r_{\widearc{\xi_{\bB}}}(\pi+\theta) \ge r_{\overline{A_0A_1}}(\theta) + r_{\overline{B_0B_1}}(\pi +\theta)$.
Note that if 
$r_{\widearc{\xi_{\bA}}}(\theta)
+r_{\widearc{\xi_{\bB}}}(\pi+\theta) > 
r_{\overline{A_0A_1}}(\theta) + r_{\overline{B_0B_1}}(\pi +\theta)$, 
it must be the case that either 
$r_{\widearc{\xi_{\bA}}}(\theta)$
is determined by a point other than $A_0$ or $A_1$,
or
$r_{\widearc{\xi_{\bB}}}(\pi+\theta)$ is determined by a point other than $B_0$ or $B_1$.
\end{obs}

The observations above allow us to give a lower bound on the length of optimal co-motions.

\begin{lem}
\label{lem:integral}
Suppose that $m$ is any collision-free net counter-clockwise co-motion from 
configuration $\bA\bB_0$ to configuration $\bA\bB_1$.
Then 
\begin{equation*}
\ell(m) \ge  \int_{S^1} \max\left(r_{\overline{A_0A_1}}(\theta) + r_{\overline{B_0B_1}}(\pi +\theta), 
r_{\bA\!+\!\bB}(\theta)
\cdot\mathds{1}_{[\theta_0, \theta_1]}(\theta) \right) \textrm{d} \theta - |\overline{A_0A_1}| - |\overline{B_0B_1}|,
\end{equation*}
where 
$\mathds{1}_{[\theta_0, \theta_1]}$ is the 
indicator function of the range $[\theta_0, \theta_1]$. \\
\elide{ SKIP THIS: Similarly, if $m$ is net clockwise, then
\begin{equation*}
\ell(m) \ge  \int_{S^1} \max\left(r_{\overline{A_0A_1}}(\theta) + r_{\overline{B_0B_1}}(\pi +\theta), 
r_{\bA\!+\!\bB}(\theta)
\cdot\mathds{1}_{S^1 - [\theta_0, \theta_1]}(\theta) \right) \textrm{d} \theta - |\overline{A_0A_1}| - |\overline{B_0B_1}|.
\end{equation*}
}
\end{lem}

\begin{defn}
Let $m = (\xi_{\bA}, \xi_{\bB})$ be any co-motion from 
configuration $\bA\bB_0$ to configuration $\bA\bB_1$.
The pair $(\widearc{\xi_{\bA}}, \widearc{\xi_{\bB}})$ 
is said to be \emph{counter-clockwise tight (with respect to the pair $(\bA\bB_0, \bA\bB_1)$)} 
if for all $\theta \in S^1$,
\begin{equation}\label{eqn:CCtight}
 r_{\widearc{\xi_{\bA}}}(\theta)
+r_{\widearc{\xi_{\bB}}}(\pi+\theta)=\max\left(r_{\overline{A_0A_1}}(\theta) + r_{\overline{B_0B_1}}(\pi +\theta), 
r_{\bA\!+\!\bB}(\theta)
\cdot\mathds{1}_{[\theta_0, \theta_1]}(\theta) \right), 
\end{equation}
\elide{ SKIP THIS: Similarly, 
$(\widearc{\xi_{\bA}}, \widearc{\xi_{\bB}})$
is said to be \emph{clockwise tight}
if for all $\theta \in S^1$,
\begin{equation}\label{eqn:Ctight}
 r_{\widearc{\xi_{\bA}}}(\theta)
+r_{\widearc{\xi_{\bB}}}(\pi+\theta)=\max\left(r_{\overline{A_0A_1}}(\theta) + r_{\overline{B_0B_1}}(\pi +\theta), 
r_{\bA\!+\!\bB}(\theta)
\cdot\mathds{1}_{S^1-[\theta_0, \theta_1]}(\theta) \right), 
\end{equation}
}
\end{defn}

\begin{obs}\label{obs:tightness} 
Let $m = (\xi_{\bA}, \xi_{\bB})$ be any co-motion from 
configuration $\bA\bB_0$ to configuration $\bA\bB_1$.
Note that the 
pair $(\widearc{\xi_{\bA}}, \widearc{\xi_{\bB}})$ 
is counter-clockwise tight if and only if for all angles $\theta$ for which $r_{\widearc{\xi_{\bA}}}(\theta)
+r_{\widearc{\xi_{\bB}}}(\pi+\theta) > r_{\overline{A_0A_1}}(\theta) + r_{\overline{B_0B_1}}(\pi +\theta)$, 
$\theta \in [\theta_0, \theta_1]$ and
$r_{\widearc{\xi_{\bA}}}(\theta)
+r_{\widearc{\xi_{\bB}}}(\pi+\theta) = r_{\bA\!+\!\bB}(\theta)$
\end{obs}

We are now in a position to state our basic tool for establishing the optimality of co-motions:

\begin{lem}
\label{lem:optkey}
Let $m = (\xi_{\bA}, \xi_{\bB})$ be any 
collision-free co-motion from initial configuration $\bA\bB_0$ to goal configuration $\bA\bB_1$.
If $m$ is net counter-clockwise and convex, and
the pair $(\widearc{\xi_{\bA}}, \widearc{\xi_{\bB}})$ 
is counter-clockwise tight
then $m$ is counter-clockwise optimal. 
\end{lem}

\begin{proof}
Let $m' = ({\xi}'_{\bA}, {\xi}'_{\bB})$ be any collision-free net counter-clockwise co-motion from $\bA\bB_0$ to $\bA\bB_1$.
It follows that
\begin{align*}
\ell(m) 
&= \ell(\xi_{\bA}) + \ell(\xi_{\bB}) \\
&= \ell(\widearc{\xi_{\bA}}) + \ell(\widearc{\xi_{\bB}}) - |\overline{A_0A_1}| - |\overline{B_0B_1}|,  \mbox{   by (1)} \\
&= \int_{S^1} \left(r_{\widearc{\xi_{\bA}}}(\theta)
+r_{\widearc{\xi_{\bB}}}(\pi+\theta) \right) 
\textrm{d} \theta -
|\overline{A_0A_1}| - |\overline{B_0B_1}|,  
\mbox{  by Corollary~\ref{cor:cauchyorig}} \\
&= \int_{S^1} \max\left(r_{\overline{A_0A_1}}(\theta) + r_{\overline{B_0B_1}}(\pi +\theta), 
r_{\bA\!+\!\bB}(\theta)
\cdot\mathds{1}_{[\theta_0,\theta_1]}(\theta) \right) -
|\overline{A_0A_1}| - |\overline{B_0B_1}|,
\mbox{  by (\ref{eqn:CCtight})}\\
&\le \ell(m'), 
\mbox{    by Lemma~\ref{lem:integral}}.
\end{align*}

\end{proof}

\begin{figure}[ht]
\centering
\includegraphics[width=0.9\textwidth]
{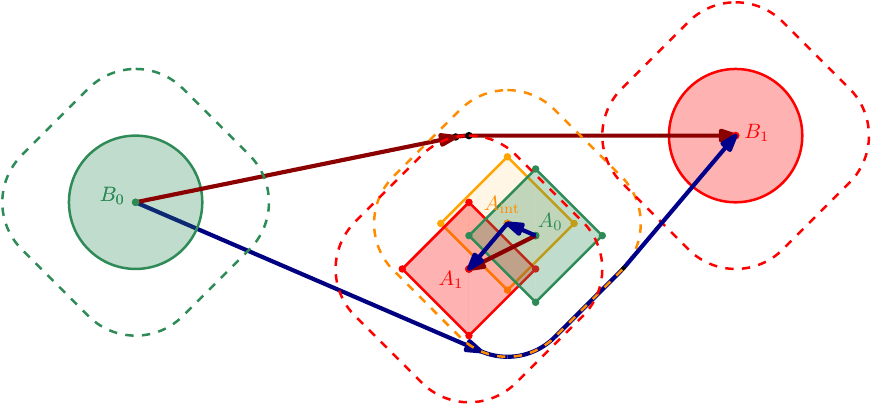}
\caption{The trace of two collision-free co-motions from the configuration $(\bA[A_0], \bB[B_0)]$ (green) to the configuration $(\bA[A_1], \bB[B_1)]$ (red). 
The dark blue co-motion is net counter-clockwise and the dark red co-motion is net clockwise.}\label{fig:introexample2}

\end{figure}

Figure~\ref{fig:introexample2} illustrates two co-motions from the configuration $(\bA[A_0], \bB[B_0)]$ (green) to the configuration $(\bA[A_1], \bB[B_1)]$ (red). 
The blue co-motion consists of three sub-motions: (i) first $\bA$ translates from placement  
$\bA[A_0]$ to placement
$\bA[A_{\rm int}]$ (yellow);
(ii) next $\bB$ translates from placement $\bB[B_0]$ to placement
$\bB[B_{1}]$, avoiding 
$\bA[A_{\rm int}]$; 
(iii) finally $\bA$ translates from placement  
$\bA[A_{\rm int}]$ to placement
$\bA[A_{1}]$.
The red co-motion consists of just two sub-motions: (i) first $\bA$ translates from placement  
$\bA[A_0]$ to placement
$\bA[A_{1}]$; (ii) next 
$\bB$ moves from placement $\bB[B_0]$ to placement
$\bB[B_{1}]$, avoiding placement
$\bA[A_{1}]$.
As we will see, the blue co-motion is 
counter-clockwise optimal and the red co-motion is clockwise optimal. Given this, it is apparent (in this case) that the clockwise optimal co-motion is globally optimal. 

In the next sections we give constructions describing  optimal co-motions for a pair $(\bA, \bB)$ of robots, from an arbitrary initial configuration
$\bA\bB_0$ to an arbitrary goal configuration $\bA\bB_1$.
This includes, of course, all those co-motions whose associated motions 
have a straight trace,
what we refer to as
\emph{straight-line co-motions}. 
In some cases globally optimal co-motions are described explicitly; in other cases, we construct both clockwise and counter-clockwise optimal co-motions, the shorter of which must be optimal among all co-motions.


\section{Optimal co-motions for two symmetric convex robots}
\label{sec:mainprf}

\subsection{An example of a counter-clockwise optimal co-motion}
\label{sec:special}

\begin{figure}[ht]
\centering
\includegraphics[width=0.9\textwidth]
{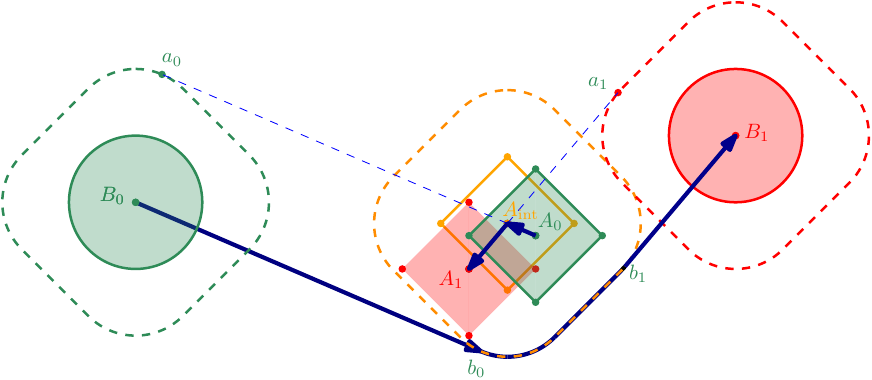}
\caption{The trace of an optimal counter-clockwise co-motion from the configuration $\bA\bB_0$ (green) to the configuration $\bA\bB_1$ (red).} \label{fig:introexample3}
\end{figure}

\begin{figure}[ht]
\centering
\includegraphics[width=0.9\textwidth]
{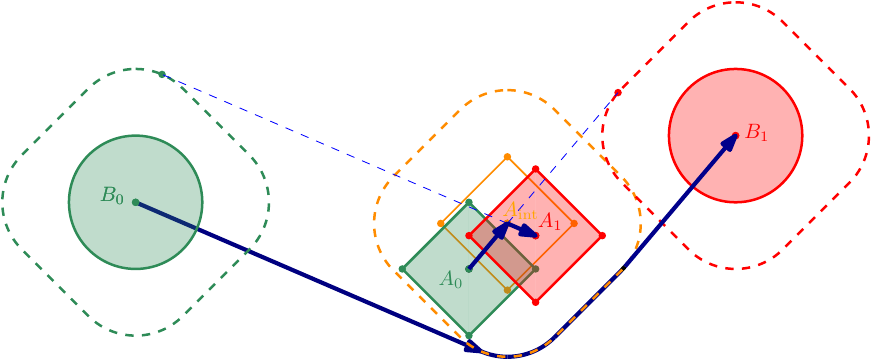}
\caption{The trace of another optimal counter-clockwise  co-motion from the configuration $\bA\bB_0$ (green) to the configuration $\bA\bB_1$ (red)}. \label{fig:introexample4}
\end{figure}

Before proceeding to classify all optimal co-motions for the pair $(\bA, \bB)$, it will be instructive to examine a special case, illustrated in Figure~\ref{fig:introexample3}, one of the co-motions illustrated previously in Figure~\ref{fig:introexample2}. 
This case will provide the simplest non-trivial example of a counter-clockwise optimal (though not globally optimal) co-motion, as well as an illustration of the form of the proofs that support our claims of optimality.

We define some points useful to our construction of the optimal co-motion. 
Let $a_0$ and $a_1$ be the upper tangent points from $A_0$ to 
$(\bA\!+\!\bB)[B_0]$, and $A_1$ to 
$(\bA\!+\!\bB)[B_1]$ respectively. 
These two tangents intersect in a point $A_\text{int}$. 
Let $b_0$ and $b_1$ be the lower tangent points from $B_0$ and $B_1$ to 
$(\bA\!+\!\bB)[A_\text{int}]$ respectively.
Note that by construction, the tangent from $A_0$ to $a_0$ passes through $A_\text{int}$ and by symmetry is parallel to the tangent from $B_0$ to $b_0$
(Similarly, the tangent from $A_1$ to $a_1$ is parallel to the tangent from $B_1$ to $b_1$.)

\begin{clm}
\label{claim:special}
The following is a counter-clockwise optimal co-motion from initial configuration $\bA\bB_0$ to goal configuration $\bA\bB_1$ (see bolded outline in Figure~\ref{fig:introexample3}):

\begin{enumerate}
\item Translate $\bA$ from position $A_0$ to position $A_\text{int}$;
\item Move $\bB$ from position $B_0$ to $B_1$, avoiding $(\bA\!+\!\bB)[A_\text{int}]$. This involves (i) translating $\bB$ from position $B_0$ to $b_0$, (ii) moving $\bB$ around the boundary of $(\bA\!+\!\bB)[A_\text{int}]$ from position $b_0$ to $b_1$, and
(iii) translating $\bB$ from position $b_1$ to $B_1$.
Then
\item Translate $\bA$ from position $A_\text{int}$ to $A_1$.
\end{enumerate}
\end{clm}

\begin{proof}
Let $\xi_{\bA}$ (respectively $\xi_{\bB}$) denote the motion of $\bA$ (respectively $\bB$), and 
$m = (\xi_{\bA}, \xi_{\bB})$.
It is straightforward to confirm that $m$ is a collision-free, convex and 
strictly counter-clockwise co-motion.
By Lemma~\ref{lem:optkey}, it remains to show that 
the pair $(\widearc{\xi_{\bA}}, \widearc{\xi_{\bB}})$
is counter-clockwise tight.

Let $\widehat{I_m}$ 
denote the range of 
orientations realized by $m$ in step 2 (i.e. counter-clockwise between the orientation of configuration 
$(\bA[A_{\rm int}], \bB[b_0])$ and the
orientation of configuration 
$(\bA[A_{\rm int}], \bB[b_1])$).
For angles  $\theta$ in $S^1 - \widehat{I_m}$, $r_{\widearc{\xi_{\bA}}}(\theta)$
is determined by one of the points $A_0$ or $A_1$,
and
$r_{\widearc{\xi_{\bB}}}(\pi+\theta)$ is determined by one of the points $B_0$ or $B_1$, 
and hence
$r_{\widearc{\xi_{\bA}}}(\theta)
+r_{\widearc{\xi_{\bB}}}(\pi+\theta) = r_{\overline{A_0A_1}}(\theta) + r_{\overline{B_0B_1}}(\pi +\theta)$.

For angles $\theta \in \widehat{I_m}$, 
$r_{\widearc{\xi_{\bA}}}(\theta)
+r_{\widearc{\xi_{\bB}}}(\pi+\theta)$
is determined by $A_{\rm int}$ and a support point on $(\bA\!+\!\bB)[A_{\rm int}]$ and hence
$r_{\widearc{\xi_{\bA}}}(\theta)
+r_{\widearc{\xi_{\bB}}}(\pi+\theta) = 
r_{\bA\!+\!\bB}(\theta)$.
Thus, for all $\theta \in S^1$, 
\begin{align*}
r_{\widearc{\xi_{\bA}}}(\theta)
+r_{\widearc{\xi_{\bB}}}(\pi+\theta) 
&= \max\left(r_{\overline{A_0A_1}}(\theta) + r_{\overline{B_0B_1}}(\pi +\theta), 
r_{\bA\!+\!\bB}(\theta)
\cdot\mathds{1}_{\widehat{I_m}}(\theta) \right)\\
&= \max\left(r_{\overline{A_0A_1}}(\theta) + r_{\overline{B_0B_1}}(\pi +\theta), 
r_{\bA\!+\!\bB}(\theta)
\cdot\mathds{1}_{I_m}(\theta) \right).
\end{align*}

\end{proof}

\begin{rem}
\label{rem:strategy}
The proof that 
the pair $(\widearc{\xi_{\bA}}, \widearc{\xi_{\bB}})$
associated with the
co-motion $m = (\xi_{\bA}, \xi_{\bB})$ in the example above 
is counter-clockwise tight, 
follows a pattern, 
supported by Observation~\ref{obs:tightness},
that we will 
use throughout
in our consideration of  convex and net counter-clockwise co-motions $m$, taking 
initial configuration
$\bA\bB_0$ to goal configuration $\bA\bB_1$.
Specifically, observe that, 
$r_{\widearc{\xi_{\bA}}}(\theta)$
is determined by either $A_0$ or $A_1$,
and $r_{\widearc{\xi_{\bB}}}(\pi+\theta)$ is determined by either $B_0$ or $B_1$ 
(and hence $r_{\widearc{\xi_{\bA}}}(\theta)
+r_{\widearc{\xi_{\bB}}}(\pi+\theta) = 
r_{\overline{A_0A_1}}(\theta) + r_{\overline{B_0B_1}}(\pi +\theta)$)
except for angles $\theta$ (all of which lie in the range $[\theta_0, \theta_1]$),
where $r_{\widearc{\xi_{\bA}}}(\theta)
+r_{\widearc{\xi_{\bB}}}(\pi+\theta) = r_{\bA\!+\!\bB}(\theta)$.
\end{rem}

\begin{obs}
\label{rem:reversal}
If the positions of $A_0$ and $A_1$ are interchanged in the co-motion $m$ of Figure~\ref{fig:introexample3}, 
there is a co-motion $m'$, specified using exactly the same steps as described in Claim~\ref{claim:special}, that has the identical  counter-clockwise tight trace (see Figure~\ref{fig:introexample4}). 
In general, this interchange of endpoints could result in a motion that is not collision-free, but whose trace is nevertheless counter-clockwise tight.
Fortunately, in such situations we can show that the globally optimal co-motion is net clockwise.
\end{obs}


\subsection{Certifying non-optimality of certain counter-clockwise co-motions}

As we have noted, the proofs we use in our case analysis of counter-clockwise optimal co-motions will largely resemble the special case discussed in Remark~\ref{rem:strategy}.
Nevertheless, as we have also noted, for certain initial and goal configurations, the strategy succeeds only in identifying a co-motion $m$ that is not collision-free, but whose trace is net counter-clockwise tight.

It turns out that, in all such situations, we can use $m$ to construct a
collision-free \emph{clockwise} co-motion $m'$, 
whose length $\ell(m')$ satisfies $\ell(m') < \ell(m)$. 
It follows that the globally optimal co-motion must be realized by a clockwise co-motion in such situations.
 
\begin{figure}[ht]
\centering
\includegraphics
{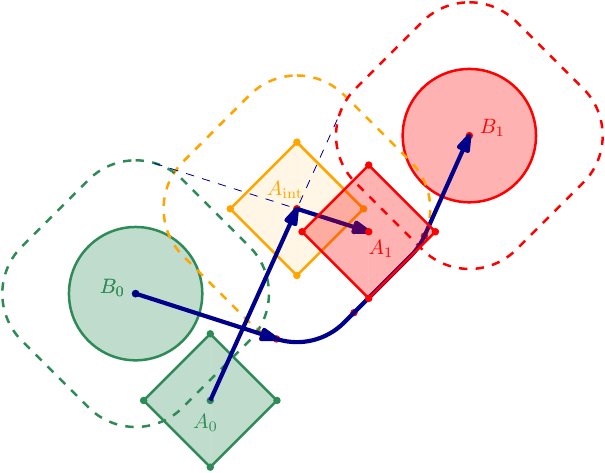}
\caption{The trace of counter-clockwise co-motion
that has no collision-free realization.}  \label{fig:introexample5}
\end{figure}

\begin{figure}[ht]
\centering
\includegraphics
{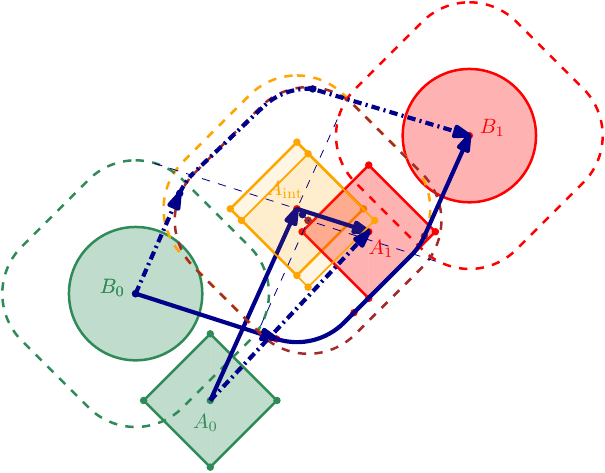}
\caption{The trace of a shorter collision-free clockwise co-motion (dashed)} \label{fig:introexample6}
\end{figure}

For example, consider the example illustrated in Figure~\ref{fig:introexample5}, a variant of the example in Figure~\ref{fig:introexample4} where the initial and goal $\bB$ placements are closer together. 
After drawing the upper tangents from $A_0$ to $(\bA\!+\!\bB)[B_1]$ and $A_1$ to $(\bA\!+\!\bB)[B_0]$, whose intersection determines the point $A_{\rm int}$, we have a co-motion $m$ (depicted in bold blue) whose 
trace is counter-clockwise tight 
(following the same argument used with the example in Figure~\ref{fig:introexample4}).
Unfortunately the co-motion $m$ is not collision-free; in fact, no co-motion 
that is trace-equivalent to $m$ 
is collision-free, since the motion of $\bA$ from $A_0$ to $A_{\rm int}$ intersects $(\bA\!+\!\bB)[B_0]$ and the motion of $\bB$ from $B_0$ to $B_1$ intersects $(\bA\!+\!\bB)[A_0]$. 

However, consider the clockwise co-motion $m'$ (dashed) illustrated in Figure~\ref{fig:introexample6}, where\\
(i) first $\bB$ moves from $B_0$ to $B_1$ following a path that is the reflection of the path followed by $\bB$ in co-motion $m$ (solid) through the midpoint of the line segment joining $B_0$ and $B_1$, then\\
(ii) $\bA$ moves from $A_0$ to $A_1$.\\
We note that $m'$ is collision-free, since $A_0$ does not intersect $(\bA\!+\!\bB)[B]$, for any point $B$ on the trace of $\bB$ from $B_0$ to $B_1$
(the tangent through $A_0$ to $(\bA\!+\!\bB)[B_0]$ and $(\bA\!+\!\bB)[B_1]$ serve as separators), and $B_1$ does not intersect $(\bA\!+\!\bB)[A]$, for any point $A$ on the trace of $\bA$ from $A_0$ to $A_1$ (the tangent from $B_1$ to $(\bA\!+\!\bB)[A_1]$ serves as a separator).
Since $m'$ is clearly shorter than $m$, we reach the following:

\begin{clm}
The optimal co-motion for the initial/goal configurations in Figure~\ref{fig:introexample6} is net clockwise. \label{claim:special-clockwise}
\end{clm}

The idea, embodied in the example above, of transforming a motion of $\bB$, from 
$\bB[B_0]$ to $\bB[B_1]$ to another motion with the same initial and goal placements, by reflection across the midpoint of the line segment joining $B_0$ and $B_1$, exploiting the symmetry of $\bA\!+\!\bB$, makes the comparison of co-motion lengths almost trivial.

\section{Optimal co-motions in standard form}
\label{sec:modified}

We now proceed to demonstrate the existence of optimal co-motions for arbitrary initial and goal configurations that take a simple standard form, using the ideas illustrated in the examples of the preceding section. We begin by dispensing with cases for which minimum length co-motions are readily apparent. 

\subsection{Straight-line co-motions}

\begin{defn}
Let $Z_0$ and $Z_1$ be arbitrary points in the plane, and let $\bZ$ denote an arbitrary centrally-symmetric robot. We denote by $\tcorr_{\bZ}(Z_0,Z_1)$ the \emph{$\bZ$-corridor} associated with placements $\bZ[Z_0]$ and $\bZ[Z_1]$, defined to be the Minkowski sum of the line segment $\overline{Z_0 Z_1}$ and $\bZ$. 
\end{defn}

As previously observed, if $\bA$ and $\bB$ are robots and $p$ is an arbitrary point, then $(\bA\!+\!\bB)[p]$
corresponds to the locations forbidden to the centre of $\bA$ in a viable placement with $\bB[p]$ and, equivalently, the locations forbidden to the centre of $\bB$ in a viable placement with $\bA[p]$  (see Figure~\ref{fig:Minkowski1}). 
Moving $\bA\!+\!\bB$ between two points $p$ and $q$ gives a range of forbidden locations whose union is exactly $\tcorr_{\bA\!+\!\bB}(p,q)$. 
When the robots have initial and goal configurations 
$(\bA[A_0], \bB[B_0])$
and $(\bA[A_1], \bB[B_1])$ respectively, the corridors $\tcorr_{\bA\!+\!\bB}(A_0,A_1)$ and $\tcorr_{\bA\!+\!\bB}(B_0,B_1)$ play a critical role in identifying initial and goal placement pairs for which  straight-line trajectories (which are clearly optimal) are possible. 

\begin{obs}
(i) Translating $\bA$ from placement $\bA[A_0]$ to placement $\bA[A_1]$ while keeping $\bB$ fixed in placement $\bB[B]$ is a 
    collision-free motion if and only if 
    $B \not\in \tcorr_{\bA\!+\!\bB}(A_0,A_1)$; and\\
    (ii)  Translating $\bB$ from placement $\bB[B_0]$ to placement $\bB[B_1]$ while keeping $\bA$ fixed in placement $\bA[A]$ is a 
    collision-free motion if and only if 
    $A \not\in \tcorr_{\bA\!+\!\bB}(B_0,B_1)$.
 \end{obs}  

It follows immediately that if 
$A_0 \not\in \tcorr_{\bA\!+\!\bB}(B_0,B_1)$
and $B_1 \not\in \tcorr_{\bA\!+\!\bB}(A_0,A_1)$ then the straight-line co-motion that first takes $\bB$ from placement $\bB[B_0]$ to placement $\bB[B_1]$, and then takes $\bA$ from placement $\bA[A_0]$ to placement $\bA[A_1]$ is feasible.
Exchanging the roles of $\bA$ and $\bB$
gives another sufficient condition for the existence of a feasible straight-line co-motion. Not surprisingly, if neither of these conditions hold then a feasible straight-line co-motion is not possible, a conclusion that follows directly from our analysis of all other cases. 

\begin{rem}
Hereafter, we simplify our illustrations, depicting $\bA\!+\!\bB$ as a disc. This special case, arising when both $\bA$ and $\bB$ are discs, is designed to help the reader visualize the general case analysis and optimality arguments in their simplest setting.  This comes with no significant loss of generality since, as we detail in Section~\ref{sec:arbitraryrobots}, the only properties of discs that we use are convexity and central symmetry. 
\end{rem}

\subsection{A normal form for initial/goal robot configurations}

We will assume, with no loss of generality, that $B_0$ and $B_1$ are horizontally aligned, and that
$B_0$ lies to the left of $B_1$.
We distinguish \emph{wide} $B$-corridors, where $(\bA\!+\!\bB)[B_0]$ and $(\bA\!+\!\bB)[B_1]$ do not intersect, and \emph{narrow} $B$-corridors, where $(\bA\!+\!\bB)[B_0]$ and $(\bA\!+\!\bB)[B_1]$ intersect. 

\begin{figure}[ht]
\centering
\includegraphics[width=0.9\textwidth]{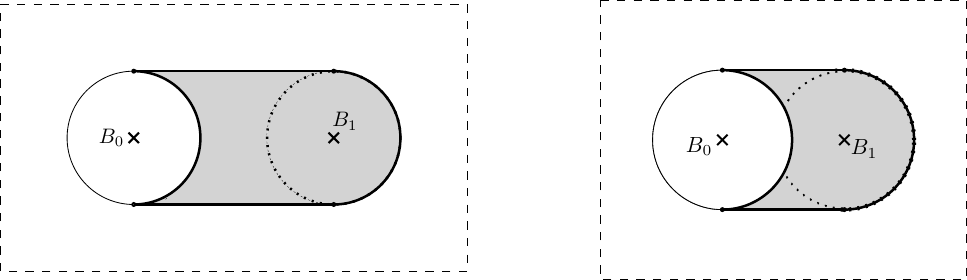}
\caption{Wide $B$-corridor (left) and narrow $B$-corridor (right). Without loss of generality $A_0 \in \tcorr(B_0,B_1) \setminus (\bA\!+\!\bB)[B_0]$ (shaded)}. \label{fig:A_0Zone}
\end{figure}

Since we will only have need to consider corridors of the form $\tcorr_{\bA\!+\!\bB}(p,q)$, we hereafter simplify notion by dropping the subscript on $\tcorr_{\bA\!+\!\bB}$. 
It is immediate from the preceding subsection that if 
(i) $A_0 \not \in \tcorr(B_0,B_1)$ and
$B_{1} \not \in \tcorr(A_0,A_1)$, or
(ii) $A_1 \not \in \tcorr(B_0,B_1)$ and
$B_{0} \not \in \tcorr(A_0,A_1)$
then a straight-line co-motion is possible. 
Hence, for the remainder of the paper, where we characterize optimal feasible co-motions when neither of these conditions hold,
we will assume, by relabelling robots and initial/goal configurations if necessary, that $A_0 \in \tcorr(B_0,B_1) \setminus (\bA\!+\!\bB)[B_0]$ and 
either (i) $A_1 \in \tcorr(B_0,B_1) \setminus (\bA\!+\!\bB)[B_1]$ or (ii) $B_0 \in \tcorr(A_0,A_1) \setminus (\bA\!+\!\bB)[A_0]$ (cf. Figure~\ref{fig:nontrivial}). 

\begin{obs} 
Note that, with these assumptions,
$\theta_0$, the orientation of $\bA\bB_0$, is in the range $(-\pi/2, \pi/2)$. Furthermore,
(i) if $A_1 \in \tcorr(B_0,B_1) \setminus (\bA\!+\!\bB)[B_1]$, then  $\theta_1$, the orientation of $\bA\bB_1$ is in the range $(\pi/2, -\pi/2)$; and
(ii) if $B_0 \in \tcorr(A_0,A_1) \setminus (\bA\!+\!\bB)[A_0]$, then $\theta_1$ corresponds to an angle
in the set
$\{ \phi_{\overrightarrow{A_0 Z}} \;|\; Z \in (\bA\!+\!\bB)[B_0] \}$
which contains $\theta_0$ in its complement.
Since the range $[\theta_0, \theta_1]$ does not contain $-\pi/2$, we refer to points on the boundary of $(\bA\!+\!\bB)[O]$ with outer normal in the range $[\theta_0, \theta_1]$ 
(respectively, $S^1 - [\theta_0, \theta_1]$, as the \emph{top} (respectively, \emph{bottom}) of  $(\bA\!+\!\bB)[O]$, with respect to the pair $(\bA\bB_0, \bA\bB_1)$.
\end{obs}

\begin{figure}[ht]
\centering
\includegraphics[width=0.9\textwidth]{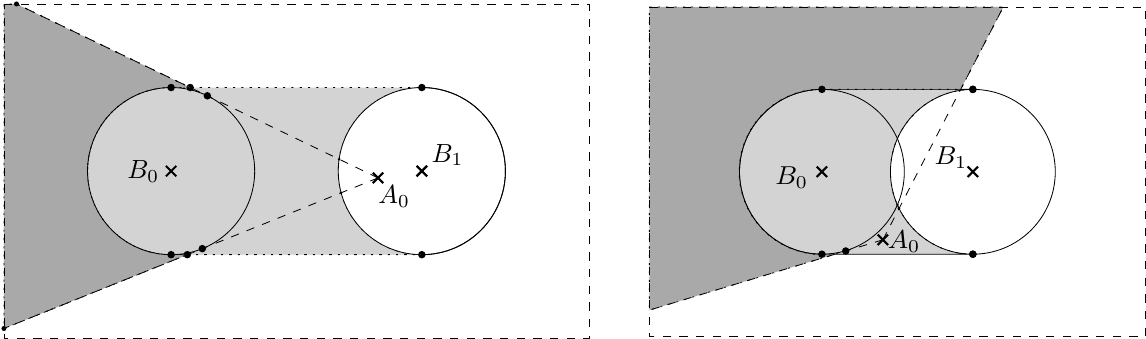}
\caption{$A_1 \in \tcorr(B_0,B_1) \setminus (\bA\!+\!\bB)[B_1]$ (light grey) or 
$B_0 \in \tcorr(A_0,A_1)$ (dark grey)}. \label{fig:nontrivial}
\end{figure}

At this point we begin to focus on counter-clockwise co-motions; the identical analysis applied to configurations formed by reflecting placements across the line through $B_0$ and $B_1$ treats the case of clockwise co-motions.

\subsection{Standard-form co-motions}
The counter-clockwise co-motions
considered so far, including straight-line co-motions, all have a three-step  (any of which may be degenerate) form that we refer to as \emph{standard-form} counter-clockwise co-motions:

\begin{enumerate}
\item Move $\bA$ from position $A_0$ to some intermediate position $A_\text{int}$, along a counter-clockwise oriented shortest path that avoids $\bB[B_0]$ (equivalently, move the centre of $\bA$ along along a counter-clockwise oriented shortest path, from $A_0$ to $A_\text{int}$, that avoids $(\bA\!+\!\bB)[B_0]$);
\item Move $\bB$ from position $B_0$ to position $B_1$, along a counter-clockwise oriented shortest path that avoids $\bA[A_\text{int}]$
(equivalently, move the centre of $\bB$ along along a counter-clockwise oriented shortest path, from $B_0$ to $B_1$, that avoids $(\bA\!+\!\bB)[A_\text{int}]$), and
\item Move $\bA$ from position $A_\text{int}$ to position $A_1$, along a counter-clockwise oriented shortest path that avoids $\bB[B_1]$
(equivalently, move the centre of $\bA$ along along a counter-clockwise oriented shortest path, from $A_\text{int}$ to $A_1$, that avoids $(\bA\!+\!\bB)[B_0]$).
\end{enumerate}

\noindent
\begin{obs}
By definition, standard-form co-motions are \emph{decoupled} and collision-free. 
In each step, one robot is stationary while the other undergoes a shortest motion whose trace 
is either (a) a straight line segment or (b) 
a convex curve consisting of three pieces (the first or last of which may be degenerate): (i) a straight segment, (ii) a counter-clockwise portion of the boundary of $\bA\!+\!\bB$ centred at the position of the robot placement being avoided, and (iii) another straight segment. 
Note that, the motion in each of the three steps is necessarily convex. 
Thus, the full standard-form counter-clockwise co-motion is convex if and only if the combined motion of $\bA$ (steps 1 and 3) is convex.
\end{obs}

It turns out that in looking for globally optimal co-motions we can restrict attention to those in standard-form:

\begin{lem}\label{lem:standard}
For any given initial and goal configurations, there is a standard-form counter-clockwise co-motion $m^* = (\xi^*_{\bA}, \xi^*_{\bB})$ for which the
pair $(\widearc{\xi^*_{\bA}}, \widearc{\xi^*_{\bB}})$ 
is counter-clockwise tight.   
\end{lem}

It follows from Lemma~\ref{lem:optkey} that in those cases where the standard-form co-motion identified in Lemma~\ref{lem:standard} is convex, the co-motion is counter-clockwise optimal. 
In those cases where the identified co-motion is not convex we will see that any globally optimal co-motion must be clockwise optimal. 
Taken together, along with Observation~\ref{obs:flip}, this provides a constructive proof of the following more precise expression of Theorem~\ref{thm:mainthm}:

\begin{thm}\label{thm:standardoptimality}
  For any given initial and goal configurations, there is a globally optimal co-motion $m^* = (\xi^*_\bA, \xi^*_\bB)$ of standard form. 
  The length of $m^*$, and hence any globally optimal motion, satisfies
  $\ell(m^*) = \ell(\widearc{tr(\xi^*_\bA) - tr(\xi^*_\bB)}) - |\overline{A_0A_1}| - |\overline{B_0B_1}|$.
\end{thm}

\begin{proof}
We use, as well, Observation~\ref{obs:Cauchy2} and
the fact that the extreme points of $\widearc{\xi_\bA} - \widearc{\xi_\bB}$ coincide with those of $tr(\xi_\bA) -tr(\xi_\bB)$.   
\end{proof}

\section{The proof of Lemma~\ref{lem:standard}}\label{sec:standardproof}

In this section we present a proof of Lemma~\ref{lem:standard} through an exhaustive treatment of the different cases 
that arise from the placement of $A_0$ and $A_1$.  
For both types of $B$-corridor, 
we look at various regions that $A_1$ could occupy, beginning with cases where $A_1 \notin \tcorr(B_0,B_1)$.

Different cases are characterized by the the position of $A_1$ relative
to upper tangents through $A_0$ to $(\bA\!+\!\bB)[B_0]$ and $(\bA\!+\!\bB)[B_1]$ (cf. Figure~\ref{fig:uppertangents}).
Note that when $A_0 \in (\bA\!+\!\bB)[B_1]$, what we refer to as the upper tangent through $A_0$ to 
$(\bA\!+\!\bB)[B_1]$ first follows the upper tangent through $A_0$ to $(\bA\!+\!\bB)[B_0]$ until it intersects with the boundary of $(\bA\!+\!\bB)[B_1]$
(see the rightmost illustrations in Figure~\ref{fig:uppertangents}).
In general the shaded regions describe possible positions for $A_1$ that give rise to structurally identical co-motions.

\begin{figure}[h!]
\centering
\includegraphics[width=0.85\textwidth]{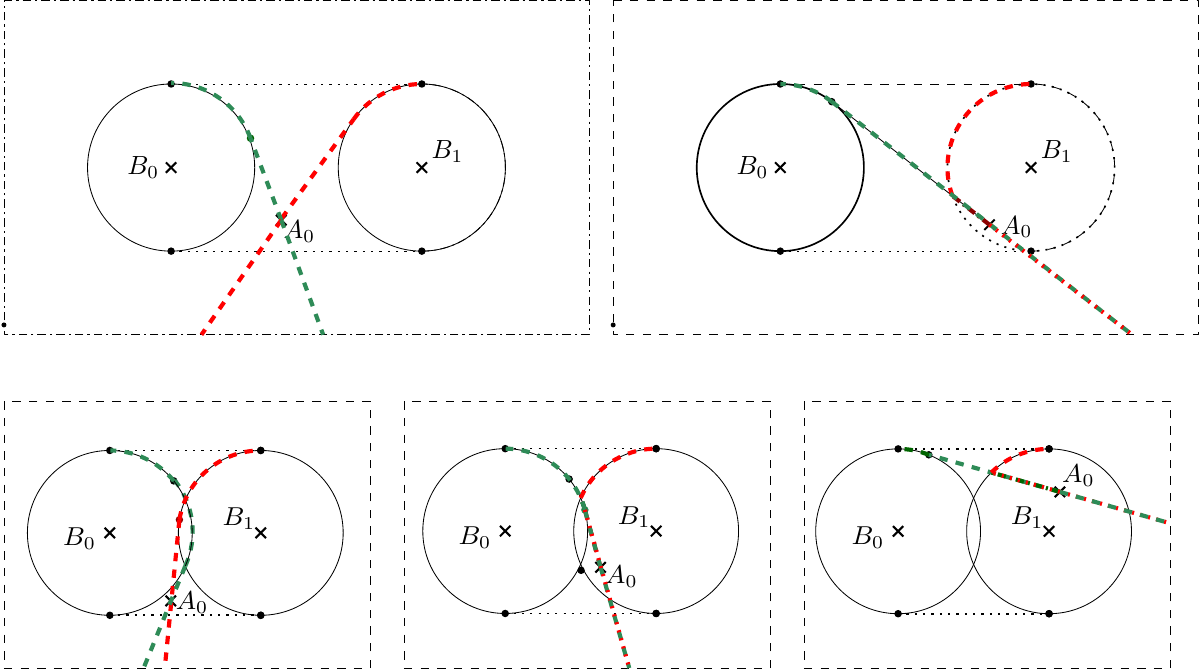}
\caption{Upper tangent through $A_0$ to $(\bA\!+\!\bB)[B_0]$ (green) and $(\bA\!+\!\bB)[B_1]$ (red).}
\label{fig:uppertangents}
\end{figure}

Unless otherwise noted, the standard co-motions that we specify have the property that the orientation of configurations changes monotonically, from the initial to the goal configuration.
Thus the counter-clockwise tightness of the trace of such standard co-motions $m = (\xi_{\bA}, \xi_{\bB})$ will follow from the observation that 
$r_{\widearc{\xi_{\bA}}}(\theta)$
is determined by either $A_0$ or $A_1$,
and
$r_{\widearc{\xi_{\bB}}}(\pi+\theta)$ is determined by either $B_0$ or $B_1$,
except for angles $\theta$ 
where either:\\
(i) $r_{\widearc{\xi_{\bB}}}(\pi +\theta)$ is determined by $B_0$ and 
$r_{\widearc{\xi_{\bA}}}(\theta)$ is determined by a point on the trace of $\xi_{\bA}$ between $A_0$ and $A_\text{int}$ that lies on the boundary of $(\bA\!+\!\bB)[B_0]$; or\\
(ii) $r_{\widearc{\xi_{\bA}}}(\theta)$ is determined by $A_\text{int}$ and 
$r_{\widearc{\xi_{\bB}}}(\pi +\theta)$ is determined by a point on the trace of $\xi_{\bB}$ between $B_0$ and $B_1$ that lies on the boundary of $(\bA\!+\!\bB)[A_\text{int}]$; or\\
(iii) $r_{\widearc{\xi_{\bB}}}(\pi +\theta)$ is determined by $B_1$ and 
$r_{\widearc{\xi_{\bA}}}(\theta)$ is determined by a point on the trace of $\xi_{\bA}$ between $A_\text{int}$ and $A_1$ that lies on the boundary of $(\bA\!+\!\bB)[B_0]$.

Given this, it suffices to specify the location of the transition point $A_\text{int}$ for all possible initial/goal configurations in our normal form.

\subsection{Cases with 
$B_0 \in \tcorr(A_0,A_1)$}

\begin{figure}[ht]
\centering
\includegraphics[width=0.9\textwidth]{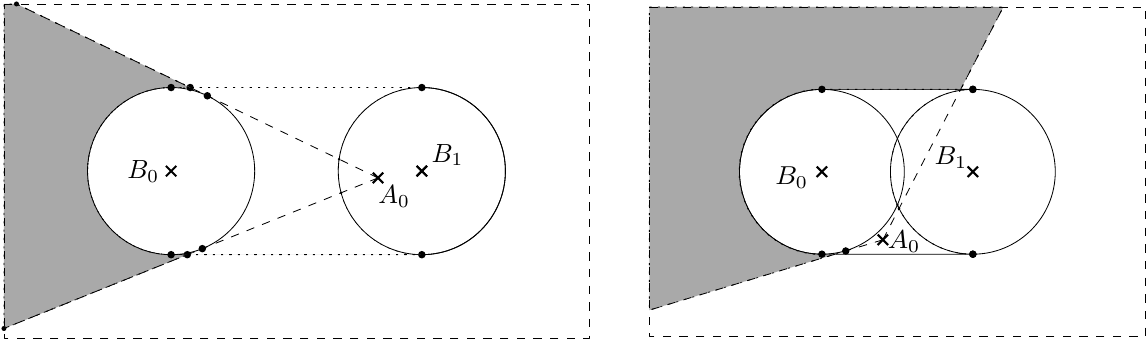}
\caption{
$B_0 \in \tcorr(A_0,A_1)$}. \label{fig:nontrivial2}
\end{figure}
\subsubsection{$A_1$ lies above the line forming the upper boundary of $\tcorr(B_0,B_1)$}

\begin{figure}[h!]
\centering
\includegraphics[width=0.9\textwidth]{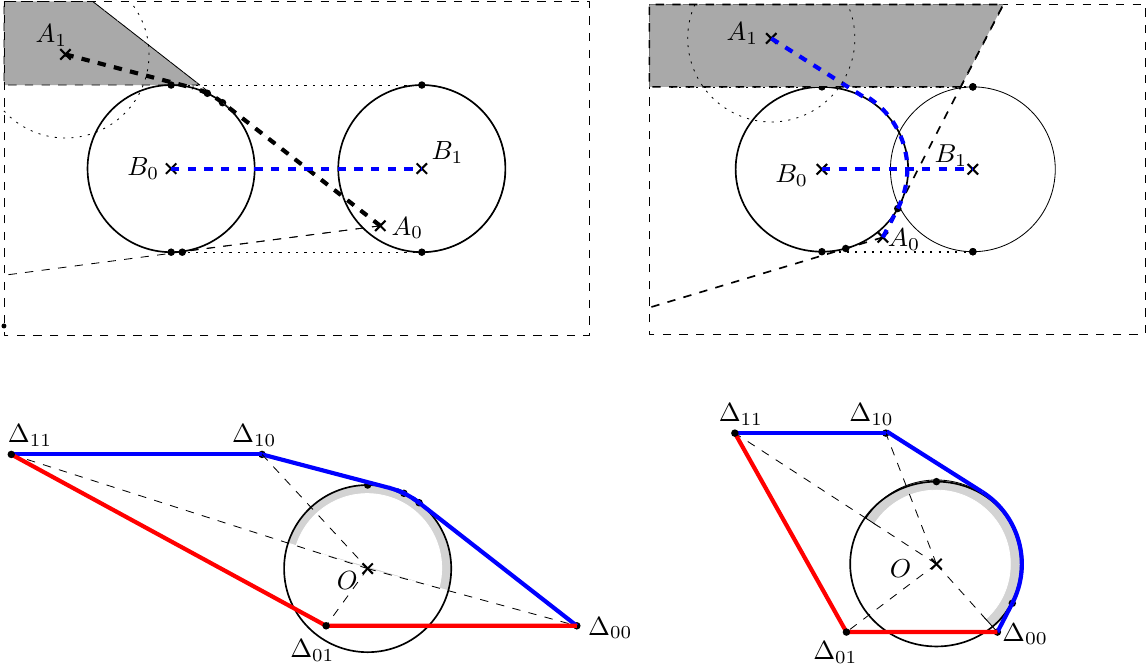}
\caption{(top) $A_1 \notin \tcorr(B_0,B_1)$: case I. (bottom) the convex hull of $(\widearc{\xi_{\bA}} -\widearc{\xi_{\bB}})$}. \label{fig:NearTrivial}
\end{figure}

In this case, illustrated in Figure~\ref{fig:NearTrivial} (top),
we take $A_\text{int}$ to be the point $A_1$. The standard-form co-motion
 $m = (\xi_{\bA}, \xi_{\bB})$ (i) first moves the centre of $\bA$ along the shortest path from $A_0$ to $A_1$ avoiding $(\bA\!+\!\bB)[B_0]$, and (ii) then moves the centre of $\bB$ along the straight path from $B_0$ to $B_1$.

Note that, except for those angles $\theta$ where 
$r_{\widearc{\xi_{\bA}}}(\theta)$ is determined by a point where the trace of $\xi_{\bA}$ coincides with a point on the boundary of $(\bA\!+\!\bB)[B_0]$,
$r_{\widearc{\xi_{\bA}}}(\theta)$ is determined by either $A_0$ or $A_1$. 
Furthermore, for all angles $\theta$,
$r_{\widearc{\xi_{\bB}}}(\pi +\theta)$
is determined by either $B_0$ or $B_1$, and in particular for those angles $\theta$ where 
$r_{\widearc{\xi_{\bA}}}(\theta)$ is determined by a point where the trace of $\xi_{\bA}$ coincides with a point on the boundary of $(\bA\!+\!\bB)[B_0]$,
$r_{\widearc{\xi_{\bB}}}(\pi +\theta)$
is determined by $B_0$. 
It follows immediately that
the pair $(\widearc{\xi_{\bA}}, \widearc{\xi_{\bB}})$ is 
counter-clockwise tight, and hence, by Lemma~\ref{lem:optkey}, $m$ is counter-clockwise optimal.

\begin{obs}
Let $\Delta_{ij}$ denote $A_i-B_j$,
and let $H^{\top}_{\bA\bB_0, \bA\bB_1}$ 
(resp. $H^{\bot}_{\bA\bB_0, \bA\bB_1}$) 
denote the convex hull of the set of points
$\{\Delta_{ij} \;|\; i,j\in \{0.1\} \}$ together with the points forming the top (resp. bottom) of $(\bA\!+\!\bB)[O]$.
It is worth noting here that, in this case,\\
(i) the convex hull of the set $(\widearc{\xi_{\bA}} -\widearc{\xi_{\bB}})$
(shown with coloured boundary in Figure~\ref{fig:NearTrivial} (bottom)
coincides with $H^{\top}_{\bA\bB_0, \bA\bB_1}$; and\\ 
(ii) the path traced by the motion $\xi_{\bA-\bB}$ (defined as $\xi_\bA(t) - \xi_\bB(t))$, shown in blue in Figure~\ref{fig:NearTrivial} (bottom), together with the edges 
$\overline{\Delta_{01} \Delta_{11}}$
and $\overline{\Delta_{00} \Delta_{01}}$
(translates of edges $\overline{A_0A_1}$ and $\overline{B_0B_1}$, shown in red), forms the boundary of $H^{\top}_{\bA\bB_0, \bA\bB_1}$.\\
A similar observation holds for all counter-clockwise optimal paths described in the following subsections, raising the question: do these properties serve as a characterization of \emph{all} counter-clockwise optimal paths? We return to this question in Section~\ref{sec:uniqueness}. 
\end{obs}

\subsubsection{$A_1$ lies 
below the upper tangent through $A_0$ to $(\bA\!+\!\bB)[B_1]$}

\begin{figure}[ht]
\centering
\includegraphics[width=0.9\textwidth]{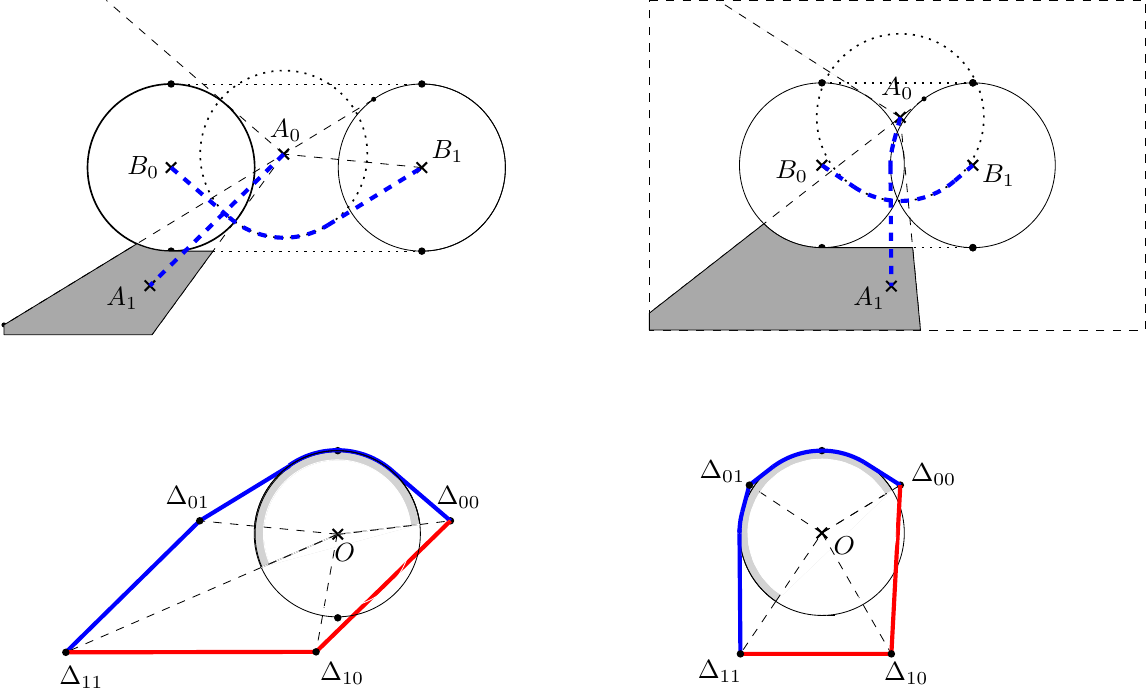}
\caption{(top) $A_1 \notin \tcorr(B_0,B_1)$:case II; (bottom) the convex hull of $(\widearc{\xi_{\bA}} -\widearc{\xi_{\bB}})$.}. \label{fig:NearTrivial4}
\end{figure}

In this case, illustrated in Figure~\ref{fig:NearTrivial4}, consider the 
standard-form co-motion $m = (\xi_{\bA}, \xi_{\bB})$ with point $A_0$ chosen as $A_\text{int}$.  Under $m$, first (i) the centre of $\bB$ moves along the shortest path from $B_0$ to $B_1$ avoiding $(\bA\!+\!\bB)[A_0]$, and then (ii) the 
centre of $\bA$ moves along the shortest path from $A_0$ to $A_1$ avoiding $(\bA\!+\!\bB)[B_1]$.

Here we note that, 
except for those angles $\theta$
where $r_{\widearc{\xi_{\bB}}}(\pi+\theta)$ is determined by a point where the trace of $\xi_{\bB}$ coincides with a point on the boundary of $(\bA\!+\!\bB)[A_0]$
(and $r_{\widearc{\xi_{\bA}}}(\theta)$ is determined by $A_0$),
$r_{\widearc{\xi_{\bB}}}(\pi +\theta)$
is determined by either $B_0$ or $B_1$.
Furthermore, except for those angles $\theta$ where 
$r_{\widearc{\xi_{\bA}}}(\theta)$ is determined by a point where the trace of $\xi_{\bA}$ coincides with a point on the boundary of $(\bA\!+\!\bB)[B_1]$
(and $r_{\widearc{\xi_{\bB}}}(\pi+\theta)$ is determined by $B_1$),
$r_{\widearc{\xi_{\bA}}}(\theta)$ is determined by either $A_0$ or $A_1$. 
It follows immediately that
the pair $(\widearc{\xi_{\bA}}, \widearc{\xi_{\bB}})$ is 
counter-clockwise tight, and hence, by Lemma~\ref{lem:optkey}, $m$ is counter-clockwise optimal.

\subsubsection{$A_1$ lies below the line forming the upper boundary of $\tcorr(B_0,B_1)$
but above the upper tangent through $A_0$ to $(\bA\!+\!\bB)[B_1]$}\label{subsub:caseIII}

\begin{figure}[ht]
\centering
\includegraphics[width=0.9\textwidth]{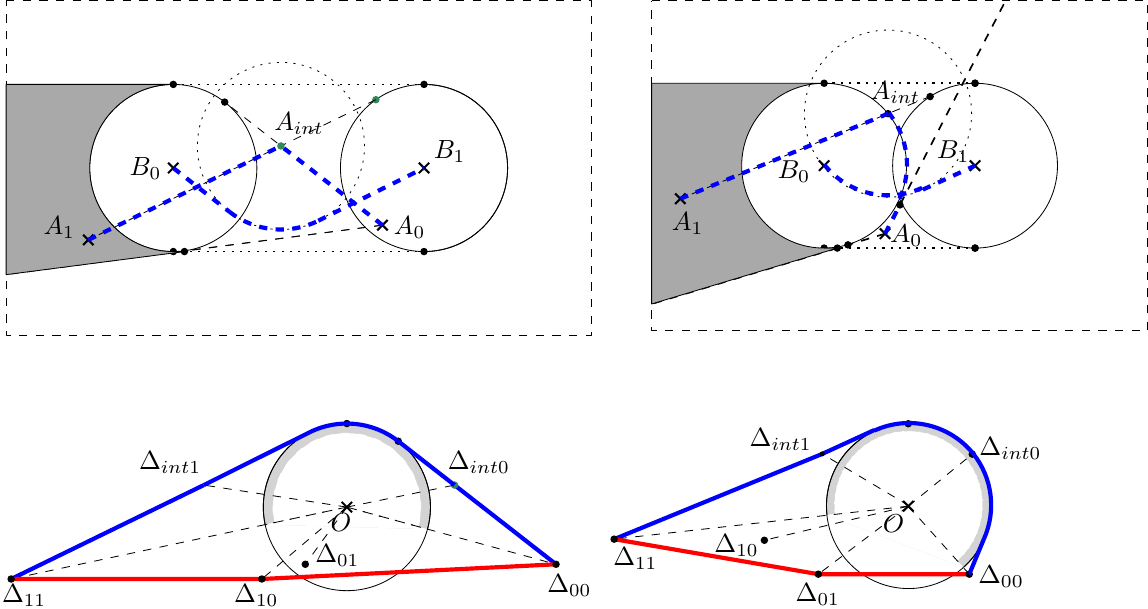}
\caption{(top) $A_1 \notin \tcorr(B_0,B_1)$: case III;
(bottom) the convex hull of $(\widearc{\xi_{\bA}} -\widearc{\xi_{\bB}})$.}. \label{fig:NearTrivial3}
\end{figure}

In this case, illustrated in Figure~\ref{fig:NearTrivial3}, let $p$ denote the point where the upper tangent through $A_0$ to $(\bA\!+\!\bB)[B_0]$
intersects the upper tangent through $A_1$ to $(\bA\!+\!\bB)[B_1]$.
Consider the 
standard-form co-motion $m = (\xi_{\bA}, \xi_{\bB})$ with point $p$ chosen as the location of $A_\text{int}$.
Under $m$, first (i) the centre of $\bA$ moves along the shortest path from $A_0$ to $p$ avoiding $(\bA\!+\!\bB)[B_0]$, then 
(ii) the centre of $\bB$ moves along the shortest path from $B_0$ to $B_1$ avoiding $(\bA\!+\!\bB)[p]$, and finally 
(iii) the centre of $\bA$ moves along the shortest path from $p$ to $A_1$ avoiding $(\bA\!+\!\bB)[B_1]$.

Here we note that
$r_{\widearc{\xi_{\bA}}}(\theta)$ is determined by either $A_0$ or $A_1$
and $r_{\widearc{\xi_{\bB}}}(\pi+\theta)$ is determined by either $B_0$ or $B_1$, except for those angles $\theta$ where 
(i) (a) $r_{\widearc{\xi_{\bA}}}(\theta)$ is determined by a point where the trace of $\xi_{\bA}$ coincides with a point on the boundary of $(\bA\!+\!\bB)[B_0]$, and 
(b) $r_{\widearc{\xi_{\bB}}}(\pi +\theta)$ is determined by $B_0$; or\\
(ii) (a) $r_{\widearc{\xi_{\bA}}}(\theta)$ is determined by the point $A_\text{int}$, and 
(b) $r_{\widearc{\xi_{\bB}}}(\pi +\theta)$ is determined by a point where the trace of $\xi_{\bB}$ coincides with a point on the boundary of $(\bA\!+\!\bB)[A_\text{int}]$
or\\
(iii) (a) $r_{\widearc{\xi_{\bA}}}(\theta)$ is determined by a point where the trace of $\xi_{\bA}$ coincides with a point on the boundary of $(\bA\!+\!\bB)[B_1]$, and 
(b) $r_{\widearc{\xi_{\bB}}}(\pi +\theta)$ is determined by $B_1$.

Points (i) and (iii) follow by the same arguments used in the preceding cases. For point (ii), we use the fact that 
(a) the tangent from $A_\text{int}$ to $(\bA\!+\!\bB)[B_0]$ coincides with the tangent from $A_0$ to $(\bA\!+\!\bB)[B_0]$, and  has the same slope as the tangent from $B_0$ to $(\bA\!+\!\bB)[A_\text{int}]$,
and (b) the tangent from $A_\text{int}$ to $(\bA\!+\!\bB)[B_1]$ coincides with the tangent from $A_1$ to $(\bA\!+\!\bB)[B_1]$, and has the same slope as the tangent from $B_1$ to $(\bA\!+\!\bB)[A_\text{int}]$.
It follows immediately that
the pair $(\widearc{\xi_{\bA}}, \widearc{\xi_{\bB}})$ is 
counter-clockwise tight, and hence, by Lemma~\ref{lem:optkey}, $m$ is counter-clockwise optimal.

\subsection{Cases with $A_1 \in \tcorr(B_0,B_1)$}

\begin{figure}[ht]
\centering
\includegraphics[width=0.9\textwidth]{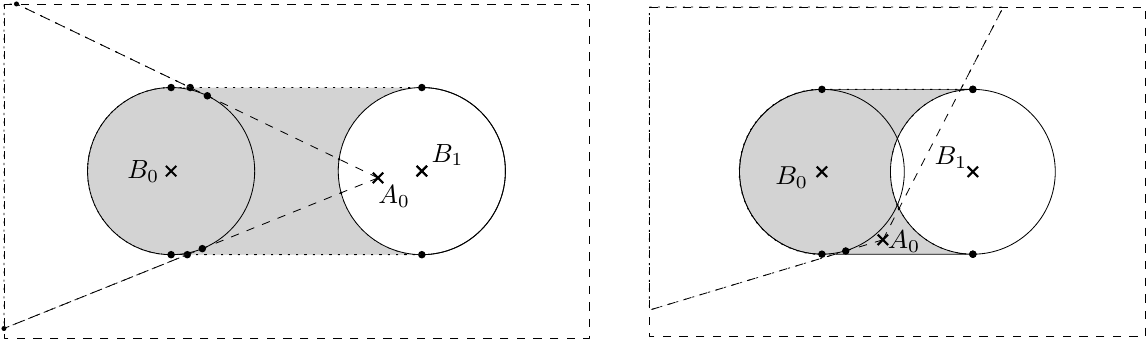}
\caption{$A_1 \in \tcorr(B_0,B_1) \setminus (\bA\!+\!\bB)[B_1]$ }. \label{fig:nontrivial3}
\end{figure}

\subsubsection{$A_1$ lies above both upper tangents through $A_0$}

\begin{figure}[h!]
\centering
\includegraphics[width=0.9\textwidth]{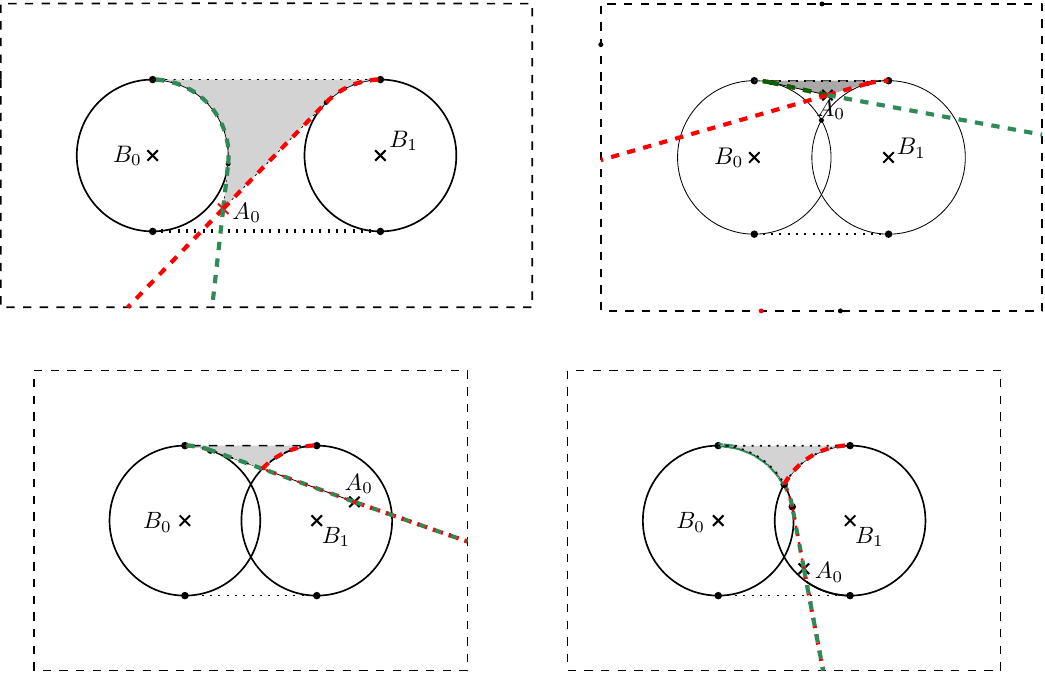}
\caption{$A_1$ lies above both upper tangents through $A_0$; $A_\text{int} = A_1$}
\label{fig:Above both}
\end{figure}

\begin{figure}[h!]
\centering
\includegraphics[width=0.9\textwidth]{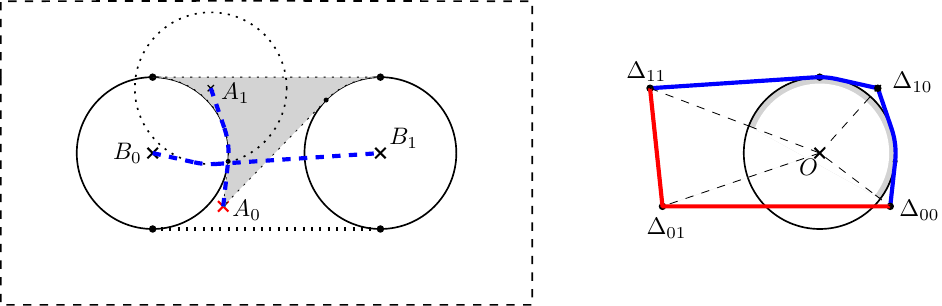}
\caption{(left) A typical co-motion when $A_1$ lies above both upper tangents through $A_0$;
(right) the convex hull of $(\widearc{\xi_{\bA}} -\widearc{\xi_{\bB}})$ in this case.}
\label{fig:Above both2}
\end{figure}

In this case, illustrated in Figure~\ref{fig:Above both}, consider the 
standard-form co-motion $m = (\xi_{\bA}, \xi_{\bB})$ with point $A_1$ chosen as the location of $A_\text{int}$.
Note that
$r_{\widearc{\xi_{\bA}}}(\theta)$
is determined by either $A_0$ or $A_1$,
and
$r_{\widearc{\xi_{\bB}}}(\pi+\theta)$ is determined by either $B_0$ or $B_1$,
except for angles $\theta$ (all of which lie in the range $[\theta_0, \theta_1]$,
where either:\\
(i) $r_{\widearc{\xi_{\bB}}}(\pi +\theta)$ is determined by $B_0$ and 
$r_{\widearc{\xi_{\bA}}}(\theta)$ is determined by a point on the trace of $\xi_{\bA}$ between $A_0$ and $A_1$ that lies on the boundary of $(\bA\!+\!\bB)[B_0]$; or\\
(ii) $r_{\widearc{\xi_{\bA}}}(\theta)$ is determined by $A_1$ and 
$r_{\widearc{\xi_{\bB}}}(\pi +\theta)$ is determined by a point on the trace of $\xi_{\bB}$ between $B_0$ and $B_1$ that lies on the boundary of $(\bA\!+\!\bB)[A_1]$.

For point (ii) we use the fact that 
the tangent from $A_1$ to $(\bA\!+\!\bB)[B_0]$ has the same slope as the tangent from $B_0$ to $(\bA\!+\!\bB)[A_1]$,
and the tangent from $A_1$ to $(\bA\!+\!\bB)[B_1]$ has the same slope as the tangent from $B_1$ to $(\bA\!+\!\bB)[A_1$. 
It follows immediately that 
the pair $(\widearc{\xi_{\bA}}, \widearc{\xi_{\bB}})$ is 
counter-clockwise tight and hence is counter-clockwise optimal, by Lemma~\ref{lem:optkey}.

\subsubsection{$A_1$ lies below both upper tangents through $A_0$}

\begin{figure}[h!]
\centering
\includegraphics[width=0.9\textwidth]{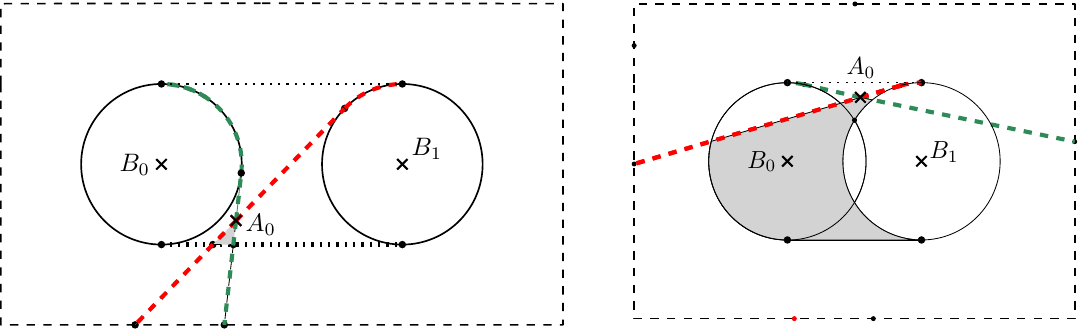}
\caption{$A_1$ lies below both upper tangents through $A_0$; $A_\text{int} = A_0$}
\label{fig:Below both}
\end{figure}

This case is illustrated in Figure~\ref{fig:Below both}. Since $A_0$ lies above both upper tangent through $A_1$, this case reduces to the preceding case, by interchanging the roles of $A_0$ and $A_1$, as well as $B_0$ and $B_1$ (noting, of course, that the reversal of any collision-free co-motion from $\bA\bB_1$ to $\bA\bB_0$ is a collision-free co-motion from  $\bA\bB_0$ to $\bA\bB_1$).
Accordingly we choose point $A_0$ as the location of $A_\text{int}$ 
in the standard-form co-motion, and the counter-clockwise optimality of this co-motion follows directly from the argument in the preceding case.

\subsubsection{$A_1$ lies left of both upper tangents through $A_0$}\label{subsubsec:leftboth}

\begin{figure}[h!]
\centering
\includegraphics[width=0.9\textwidth]{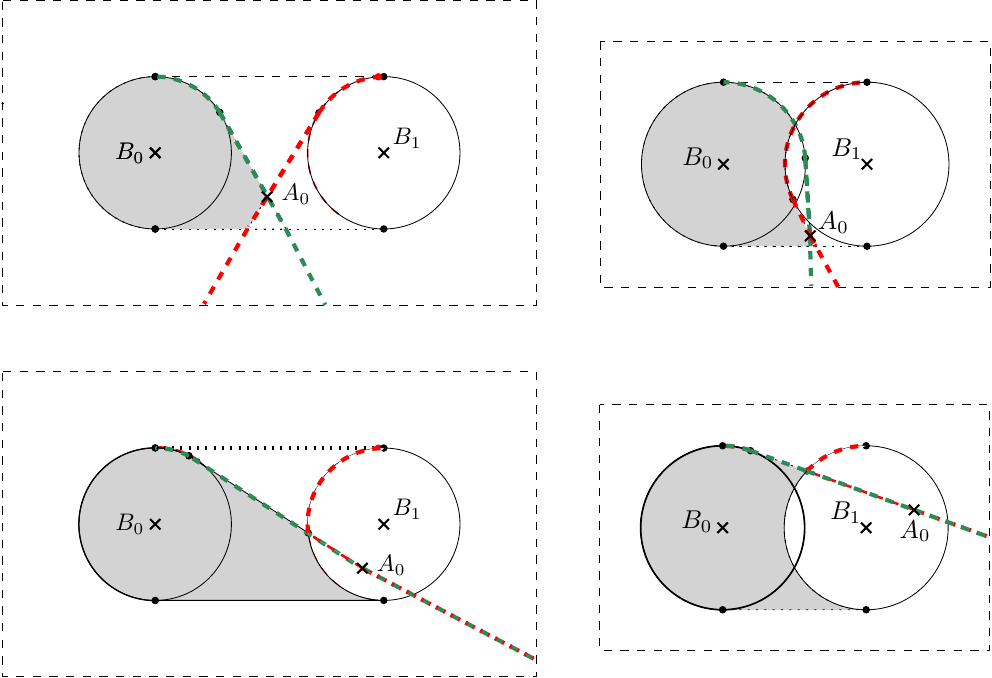}
\caption{$A_1$ lies left of both upper tangents through $A_0$; $A_\text{int}$ is the point where the tangent through $A_0$ to $(\bA\!+\!\bB)[B_0]$ crosses the tangent through $A_1$ to $(\bA\!+\!\bB)[B_1]$}
\label{fig:A1left}
\end{figure}

\begin{figure}[h!]
\centering
\includegraphics[width=.95\textwidth]{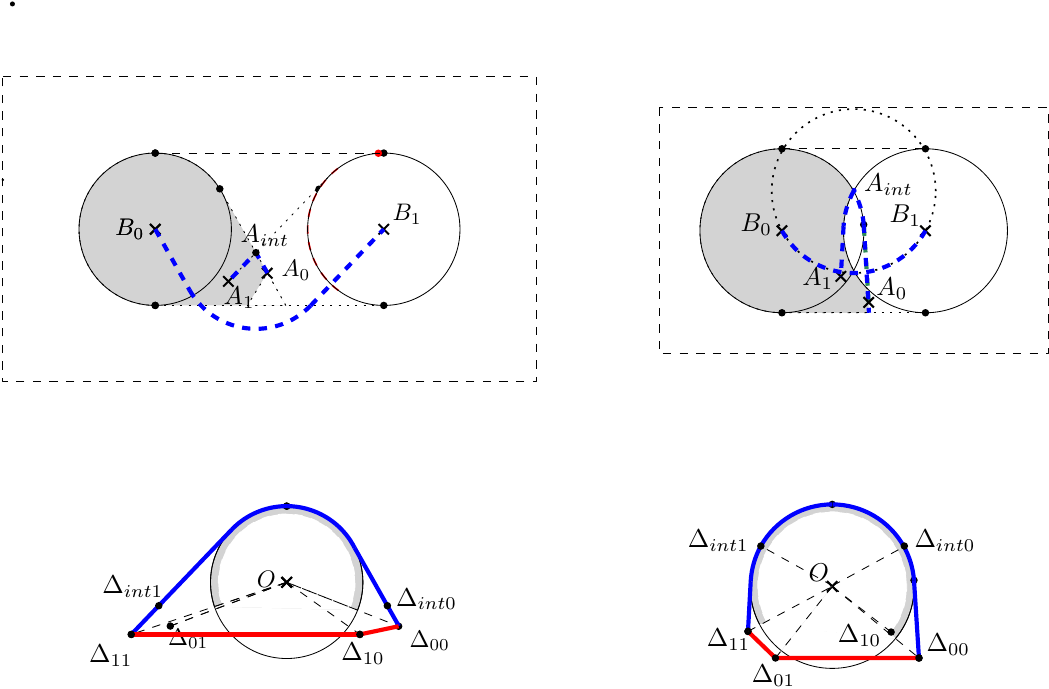}
\caption{(top) typical co-motions when $A_1$ lies left of both upper tangents through $A_0$; 
(bottom) the convex hull of $(\widearc{\xi_{\bA}} -\widearc{\xi_{\bB}})$ in these cases.}
\label{fig:A1left2}
\end{figure}

In this case, illustrated in Figure~\ref{fig:A1left}, consider the 
standard-form co-motion $m = (\xi_{\bA}, \xi_{\bB})$ with $A_\text{int}$ located at the point $p$, where the upper tangent through $A_1$ to $(\bA\!+\!\bB)[B_1]$ intersects 
the upper tangent through $A_0$ to $(\bA\!+\!\bB)[B_0]$.
The convexity of the motion of $\bA$ is apparent from inspection of the various sub-cases (shown in Figure~\ref{fig:A1left}---wide corridor (left) and narrow corridor (right), $A_0 \notin (\bA\!+\!\bB)[B_1]$ (top) and $A_0 \in (\bA\!+\!\bB)[B_1]$ (bottom)---
that determine the location of the intersection point $p$.

As we saw in the case analysed in Section~\ref{subsub:caseIII},
$r_{\widearc{\xi_{\bA}}}(\theta)$ is determined by either $A_0$ or $A_1$
and $r_{\widearc{\xi_{\bB}}}(\pi+\theta)$ is determined by either $B_0$ or $B_1$, except for those angles $\theta$ where 
(i) (a) $r_{\widearc{\xi_{\bA}}}(\theta)$ is determined by a point where the trace of $\xi_{\bA}$ coincides with a point on the boundary of $(\bA\!+\!\bB)[B_0]$, and 
(b) $r_{\widearc{\xi_{\bB}}}(\pi +\theta)$ is determined by $B_0$; or\\
(ii) (a) $r_{\widearc{\xi_{\bA}}}(\theta)$ is determined by the point $A_\text{int}$, and 
(b) $r_{\widearc{\xi_{\bB}}}(\pi +\theta)$ is determined by a point where the trace of $\xi_{\bB}$ coincides with a point on the boundary of $(\bA\!+\!\bB)[A_\text{int}]$
or\\
(iii) (a) $r_{\widearc{\xi_{\bA}}}(\theta)$ is determined by a point where the trace of $\xi_{\bA}$ coincides with a point on the boundary of $(\bA\!+\!\bB)[B_1]$, and 
(b) $r_{\widearc{\xi_{\bB}}}(\pi +\theta)$ is determined by $B_1$.
The analysis here follows in exactly the same way as that presented in Section~\ref{subsub:caseIII}

By Observation~\ref{obs:tightness}, it follows that 
the pair $(\widearc{\xi_{\bA}}, \widearc{\xi_{\bB}})$ is 
counter-clockwise tight,
and hence is counter-clockwise optimal, by Lemma~\ref{lem:optkey}.

\subsubsection{$A_1$ lies right of both upper tangents through $A_0$}\label{subsec:rightofboth}

\begin{figure}[h!]
\centering
\includegraphics[width=1.0\textwidth]{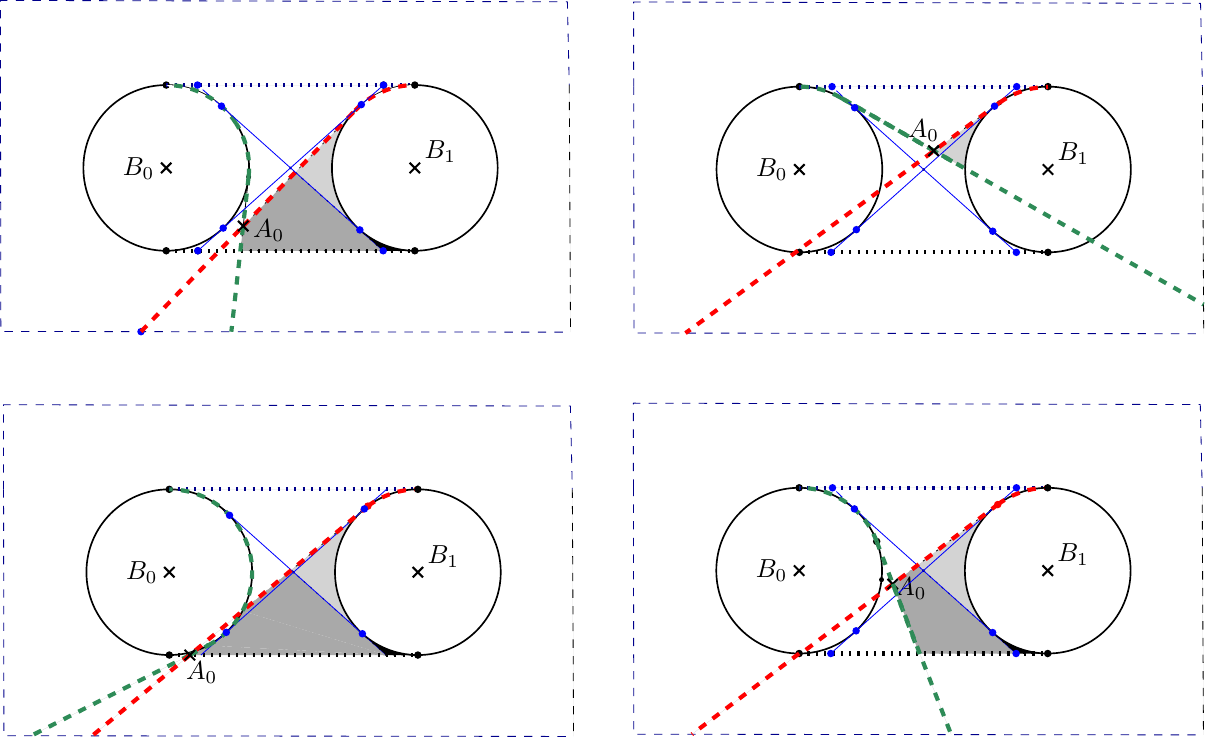}
\caption{$A_1$ lies right of both upper tangents through $A_0$}
\label{fig:A1right}
\end{figure}

\begin{figure}[h!]
\centering
\includegraphics[width=1.0\textwidth]{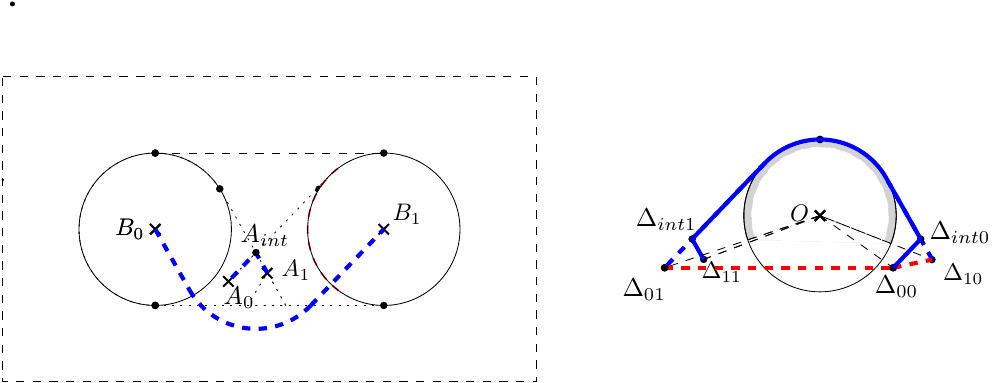}
\caption{Example when $\xi_\bA$ is convex.}
\label{fig:A1right2}
\end{figure}

In this case, illustrated in Figure~\ref{fig:A1right}, it is instructive to consider first
a standard-form co-motion $m = (\xi_\bA, \xi_\bB)$ with $A_\text{int}$ located at the point $p$, where the upper tangent through $A_1$ to $(\bA\!+\!\bB)[B_0]$ intersects 
the upper tangent through $A_0$ to $(\bA\!+\!\bB)[B_1]$.
In the event that the $\xi_\bA$, the motion of $\bA$, is convex (i.e. the first and third steps take $\bA$ along straight line segments; cf. Figure~\ref{fig:A1right2}), 
the pair $(\widearc{\xi_{\bA}}, \widearc{\xi_{\bB}})$ coincides with the corresponding pair encountered in Section~\ref{subsubsec:leftboth}. 
Thus, the counter-clockwise tightness of 
the pair $(\widearc{\xi_{\bA}}, \widearc{\xi_{\bB}})$,
and hence the counter-clockwise optimality of $m$, follows exactly as argued in the preceding case, with points $A_0$ and $A_1$ interchanged.

For the same reason, even when $\xi_\bA$ is not convex, 
the pair $(\widearc{\xi_{\bA}}, \widearc{\xi_{\bB}})$ remains 
counter-clockwise tight. (Here we use the fact that, as in the convex case, $\widearc{\xi_\bA}$ is just the triangle formed by the points $A_0$, $A_\text{int}$, and $A_1$).
Thus, by Lemma~\ref{lem:integral}, 
$\ell(\widearc{\xi_{\bA}}) + \ell(\widearc{\xi_{\bB}}) - |\overline{A_0A_1}| - |\overline{B_0B_1}|$ forms a lower bound on the 
length of any counter-clockwise collision-free co-motion.
Since either segment $A_0 A_\text{int}$ intersects $(\bA\!+\!\bB)[B_0]$, or 
segment $A_\text{int} A_1$ intersects $(\bA\!+\!\bB)[B_1]$ (or both),
we can assume, by exchanging the roles of $\bA\bB_0$ and $\bA\bB_1$ if necessary, that segment $A_0 A_\text{int}$ intersects $(\bA\!+\!\bB)[B_0]$. 
Accordingly, $A_0$ lies below the lower tangent to $(\bA\!+\!\bB)[B_0]$ through $A'_\text{int}$, 
the reflection of the point $A_\text{int}$ across the midpoint of the segment joining $B_0$ to $B_1$.

In this situation it is straightforward to construct a collision-free net clockwise co-motion $m' = (\xi'_\bA, \xi'_\bB)$, 
taking configuration $\bA\bB_0$ to configuration $\bA\bB_1$, that satisfies
$\ell(\xi'_{\bB}) = \ell(\xi_{\bB}) = \ell(\widearc{\xi_\bB}) - |\overline{B_0B_1}|$ and 
$\ell(\xi'_{\bA}) < \ell(\xi_{\bA})= \ell(\widearc{\xi_\bA}) - |\overline{A_0A_1}|$.
From this, 
it follows immediately that any globally optimal co-motion must be net clockwise.

\begin{figure}
[h!]
\centering
\includegraphics[scale=0.6]{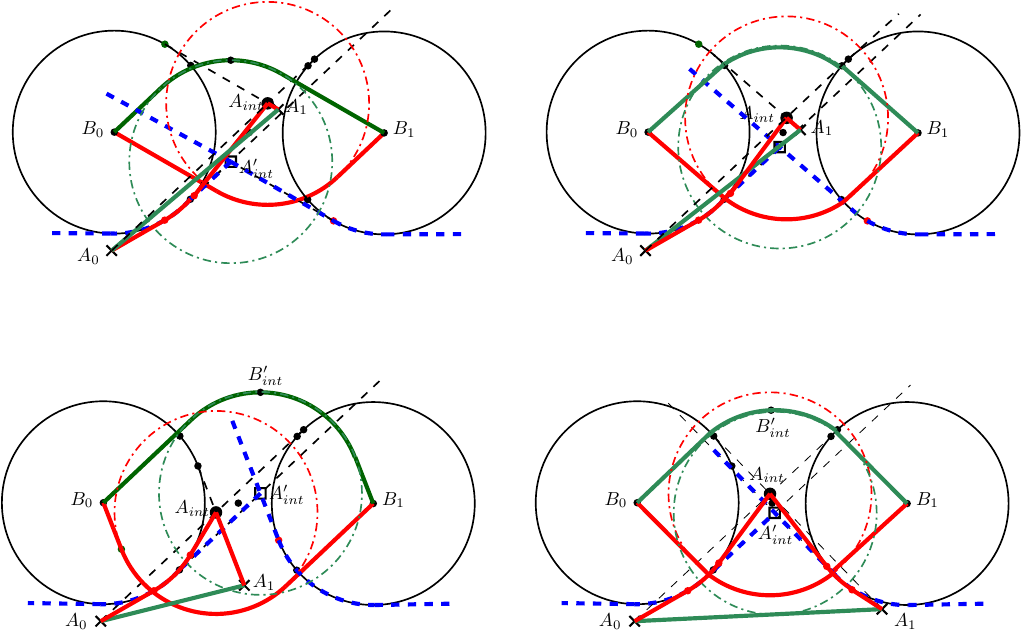}
\caption{Cases where optimal counter-clockwise co-motions are not globally optimal 
} 
\label{fig:sub-optimal}
\end{figure}

There are just two cases to consider:\\ 
(i) \emph{$A_1$ lies above the lower tangent through $A'_\text{int}$ to $(\bA\!+\!\bB)[B_1]$.}\\
In this case, illustrated in Figure~\ref{fig:sub-optimal} (top), the co-motion $m'$  first translates $\bB$ from position $B_0$ to position $B_1$ along the shortest clockwise-oriented path (shown in green) that avoids $(\bA\!+\!\bB)[A'_\text{int}]$. Note that the trace of $\xi'_\bB$ 
is the reflection, across the midpoint of the segment joining $B_0$ to $B_1$, of the trace of $\xi_\bB$
(shown in red), and hence $\ell(\xi'_\bB) = \ell(\xi_\bB) $.  Since the trace of $\xi'_\bB$ does not intersect $(\bA\!+\!\bB)[A_0]$ (from our assumption that the line segment from $A_0$ to $A_\text{int}$ intersects $(\bA\!+\!\bB)[B_0]$) it follows that this step is collision-free. 
The co-motion $m'$ completes by translating $\bA$ from position $A_0$ to position $A_1$ directly. This motion of $\bA$ is collision-free since 
$A_1$ must lie between the upper and lower tangent through $A_0$ to $(\bA\!+\!\bB)[B_1]$.  
Clearly, $\ell(\xi'_\bA) = \ell(\xi_\bA) $.

(ii) \emph{$A_1$ lies below the lower tangent through $A'_\text{int}$ to $(\bA\!+\!\bB)[B_1]$.}\\
In this case, illustrated in Figure~\ref{fig:sub-optimal} (bottom), the co-motion $m'$  first translates $\bB$ from position $B_0$ to position $B'_\text{int}$ on $(\bA\!+\!\bB)[A'_\text{int}]$, along the shortest clockwise-oriented path (shown in green) that avoids $(\bA\!+\!\bB)[A'_\text{int}]$, where the orientation of configuration $(\bA[A'_\text{int}, \bB[B'_\text{int})$ is vertical. 
Next $m'$ translates $\bA$ from position $A_0$ to position $A_1$ directly. 
Finally $m'$ translates $\bB$ from position $B'_\text{int}$ to position $B_1$, along the shortest clockwise-oriented path that avoids $(\bA\!+\!\bB)[A'_\text{int}]$.
Note that, as in the previous case, the full trace of $\xi'_\bB$ 
is the reflection, across the midpoint of the segment joining $B_0$ to $B_1$, of the trace of $\xi_\bB$
(shown in red).  
Since the trace of $\xi'_\bB$ in the first step does not intersect $(\bA\!+\!\bB)[A_0]$ (from our assumption that the line segment from $A_0$ to $A_\text{int}$ intersects $(\bA\!+\!\bB)[B_0]$),
and the trace of $\xi'_\bB$ in the third step does not intersect $(\bA\!+\!\bB)[A_1]$
(from our assumption that $A_1$ lies below the lower tangent through $A'_\text{int}$ to $(\bA\!+\!\bB)[B_1]$), it follows that the motion $\xi'_\bB$ is collision-free. Similarly the motion $\xi'_\bA$ is collision-free since the entire trace of this motion is separated from $(\bA\!+\!\bB)[B'_\text{int}]$ by the horizontal line through $A_\text{int}$.\\
As in the preceding case, 
$\ell(\xi'_\bB) = \ell(\xi_\bB) $ and
$\ell(\xi'_\bA) < \ell(\xi_\bA) $.

\subsubsection{$A_1$ lies between doubly intersecting upper tangents through $A_0$}

\begin{figure}[ht]
\centering
\includegraphics[width=0.9\textwidth]{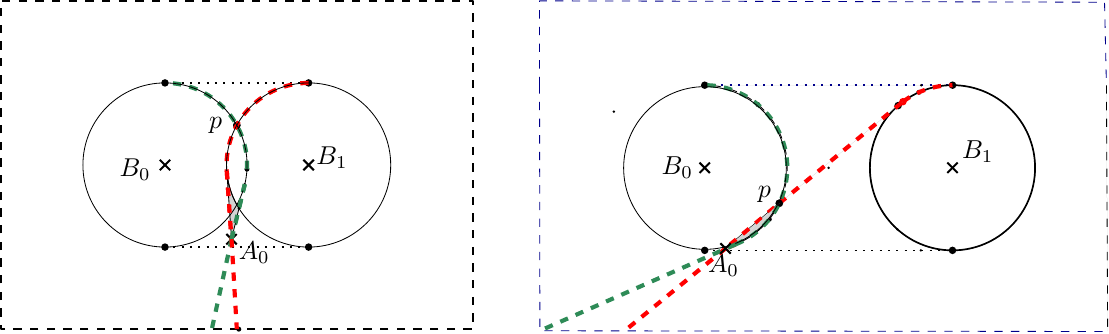}
\caption{$A_1$ lies between intersecting upper tangents through $A_0$}
\label{fig:Between tangents}
\end{figure}

\begin{figure}[ht]
\centering
\includegraphics[width=0.8\textwidth]{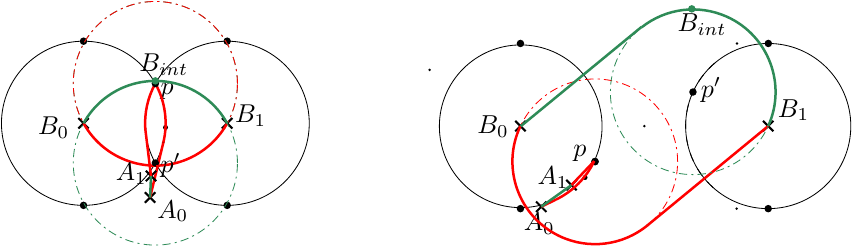}
\caption{Demonstration of non-optimality of counter-clockwise co-motion (red)}
\label{fig:Between tangents2}
\end{figure}

When $A_1$ lies between doubly intersecting upper tangent through $A_0$, illustrated in Figure~\ref{fig:Between tangents}, consider the standard-form counter-clockwise co-motion $m = (\xi_{\bA}, \xi_{\bB})$, with $A_\text{int}$ located at the point $p$, where the upper tangent through $A_0$ to $(\bA\!+\!\bB)[B_0]$ intersects 
the upper tangent through $A_0$ to $(\bA\!+\!\bB)[B_1]$ (shown in red in Figure~\ref{fig:Between tangents2}). 
As with the case described in the previous subsection, co-motion $m$
is not convex. 
In particular, 
$\widearc{\xi_{\bA}}$, 
the convex hull of the trace of $\xi_{\bA}$, contains position $A_1$ in its interior. 
Nevertheless, 
following the same argument used in the case analysed in Section~\ref{subsubsec:leftboth}, when
$A_1$ coincides with $A_0$, it is straightforward to confirm that 
the pair $(\widearc{\xi_{\bA}}, \widearc{\xi_{\bB}})$ is 
counter-clockwise tight.
Thus, by Lemma~\ref{lem:integral}, 
$\ell(\widearc{\xi_{\bA}}) + \ell(\widearc{\xi_{\bB}}) - |\overline{A_0A_1}| - |\overline{B_0B_1}|$ forms a lower bound on the 
length of any counter-clockwise collision-free co-motion from $\bA\bB_0$ to $\bA\bB_1$.

As an alternative, consider the net \emph{clockwise} 
co-motion $m'= (\xi'_{\bA}, \xi'_{\bB})$, where $\xi'_{\bB}$ is the reflection of  $\xi_{\bB}$ across the midpoint of the segment $B_0 B_1$ and $\xi'_{\bA}$ takes $\bA$ on the straight path from $A_0$ to $A_1$. 
Since $\xi'_{\bB}$ does not intersect $(\bA\!+\!\bB)[A_0]$, and the segment $A_0 A_1$ does not intersect $(\bA\!+\!\bB)[B_1]$, $m'$ is collision-free.
But, by construction, $\ell(\xi'_{\bB}) = \ell(\xi_{\bB})$
and $\ell(\xi'_{\bA}) = |\overline{A_0A_1}|$.
Furthermore, since $\overline{A_0A_1}$ forms a sudset of a chord of $\widearc{\xi_{\bA}}$, $|\overline{A_0A_1}|| < \ell(\widearc{\xi_{\bA}})/2$,
i.e. $|\overline{A_0A_1}|| < \ell(\widearc{\xi_{\bA}}) - |\overline{A_0A_1}|$.
Hence, $\ell(m')$ is less than the
length of the shortest counter-clockwise collision-free co-motion from $\bA\bB_0$ to $\bA\bB_1$.

\subsubsection{$A_1$ lies below doubly intersecting upper tangents through $A_0$}

\begin{figure}[ht]
\centering
\includegraphics[width=0.9\textwidth]{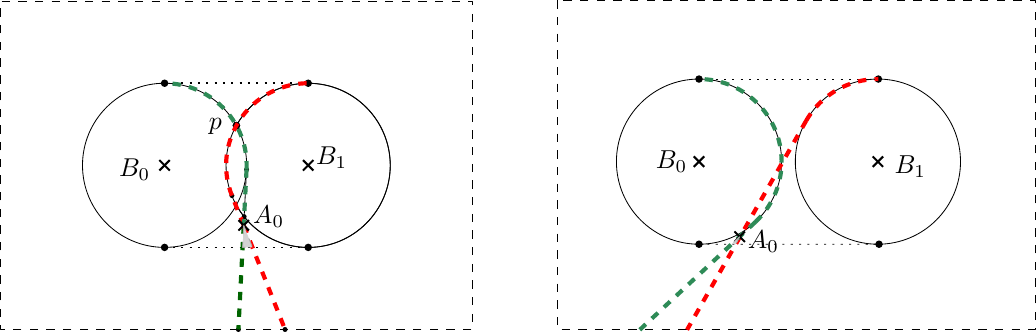}
\caption{$A_1$ lies below intersecting upper tangents through $A_0$}
\label{fig:Below intersecting tangents}
\end{figure}

Since $A_1$ lies below doubly intersecting upper tangents through $A_0$ if and only if $A_1$ lies between doubly intersecting upper tangent through $A_0$ (cf. Figure~\ref{fig:Below intersecting tangents}), this case is identical to the preceding case, after interchange of $A_0$ and $A_1$.

\subsection{Arbitrary convex and centrally-symmetric robots}\label{sec:arbitraryrobots}

For simplicity, our case analysis and optimization arguments have been illustrated using discs to describe the Minkowski sum of the participant robots. 
However, nowhere have we relied on properties of this Minkowski sum other than  convexity and central symmetry which, as we have noted, follow immediately, given that these same properties hold for the individual robots.
To be precise, we only use the fact that for such robots $\bA$ and $\bB$:\\
(i) $\bA\!+\!\bB = \bA\!-\!\bB$;\\
(ii) the upper tangent from point $A$ to $(\bA\!+\!\bB)[B]$ is parallel to the lower tangent from point $B$ to $(\bA\!+\!\bB)[A]$; and\\
(iii) the shortest path from $B_0$ to $B_1$ below $(\bA\!+\!\bB)[A]$ has the same length as the shortest path from $B_0$ to $B_1$ above $(\bA\!+\!\bB)[A']$, where $A'$ denotes the reflection of $A$ across the midpoint of the segment $\overline{B_0B_1}$ (cf. Section\ref{sec:special}).

\begin{obs}
The reader might wonder if the same or similar arguments might make it possible to further generalize our results, for example by relaxing the central symmetry constraint. 
Although our techniques allow us to identify optimal motions of arbitrary convex robots in many cases, specifically those for which we are able identify both counterclockwise optimal and clockwise optimal motions, there remain cases which seem to require new insights. 
\end{obs}

\section{Further constrained optimal co-motions}

We have seen that to minimize co-motion length it suffices to use co-motions of standard form. 
While such co-motions have certain desirable properties, their simplicity, particularly the fact that their constituent motions are decoupled, means that other useful properties do not always hold.  
In this section, we show that straightforward modifications of standard-form co-motions make it possible to attain several such additional properties.

\subsection{Orientation monotonicity}

\begin{figure}[ht]
\centering
\includegraphics[width=0.7\textwidth]{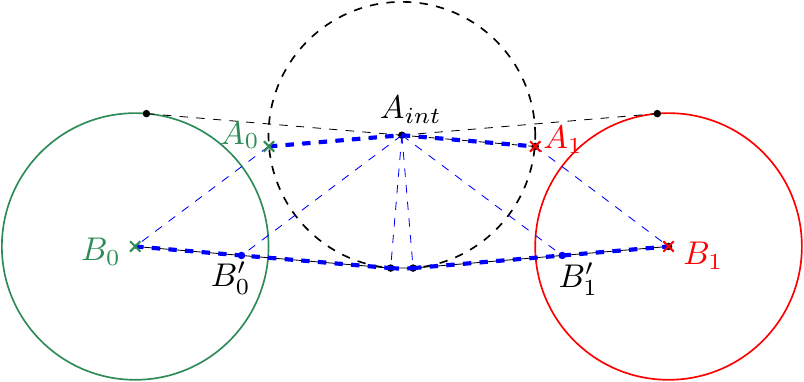}
\caption{Standard co-motion can give rise to non-monotone orientations}
\label{fig:non-monotone}
\end{figure}

By definition, standard-form co-motions are \emph{decoupled}: at any fixed time, only one of $\bA$ or $\bB$ is moving. As previously noted, in almost all cases the orientation of the associated configurations changes monotonically. However, there are situations, 
one is illustrated in Figure~\ref{fig:non-monotone}, where this is not the case. (In this situation the first and third steps result in a clockwise change in orientation.)
However, a straightforward modification to any optimal co-motion $m = (\xi_\bA, \xi_\bB)$, based on the following observation, produces a collision-free orientation-monotone co-motion with the same trace.
\begin{obs}\label{obs:monotonize}
  If configurations $(\bA[\xi_\bA(t)], \bB[\xi_\bB(t)]$ and  $(\bA[\xi_\bA(t')], \bB[\xi_\bB(t')]$ have the same orientation $\theta$, for some $t < t'$, then
  replacing sub-motion $\xi_\bA$ on the interval $[t, t']$ by the straight-line sub-motion $\xi'_\bA$ where 
  $\xi_\bA(\alpha t + (1-\alpha) t') = \alpha \xi_\bA(t) + (1-\alpha) \xi_\bA(t')$, 
  and replacing sub-motion $\xi_\bB$ on the interval $[t, t']$ by the straight-line sub-motion $\xi'_\bB$ where 
  $\xi_\bB(\alpha t + (1-\alpha) t') = \alpha \xi_\bB(t) + (1-\alpha) \xi_\bB(t')$, for $0 \le \alpha \le 1$, produces a co-motion whose configurations over the interval $[t, t']$ all have orientation $\theta$.
\end{obs}

\begin{clm}
 Any optimal co-motion $m$ can be transformed to a trace-equivalent orientation-monotone co-motion $m'$.  
\end{clm}

\begin{proof}
 It suffices to apply the modification described in Observation~\ref{obs:monotonize} to all
 maximal (time) intervals (necessarily straight co-motions, if $m$ is optimal) that repeat a configuration orientation $\theta$.
\end{proof}

In the situation illustrated in Figure~\ref{fig:non-monotone},
let $B'_0$ be the point on the trace of $\bB$ where configuration 
$(\bA[A_\text{int}], \bB[B'_0])$ has the same orientation as the initial configuration $(\bA[A_0], \bB[B_0])$, and let
$B'_1$ be the point on the trace of $\bB$ where configuration 
$(\bA[A_\text{int}], \bB[B'_1])$ has the same orientation as the goal configuration $(\bA[A_1], \bB[B_1])$. 
The modified co-motion\\
(i) starts with a straight-line co-motion that \emph{couples} the motion of $\bA$ from $A_0$ to $A_\text{int}$ with a motion 
of $\bB$ from $B_0$ to $B'_0$, preserving the orientation of the initial configuration;\\
(ii) continues with a standard co-motion from configuration $(\bA[A_\text{int}], \bB[B'_0])$ to configuration $(\bA[A_\text{int}], \bB[B'_1])$; and\\
(ii) finishes with a straight-line co-motion that \emph{couples} the motion of $\bA$
from $A_\text{int}$ to $A_1$ with a motion 
of $\bB$ from $B'_1$ to $B_1$, preserving the orientation of the goal configuration.

\subsection{Contact preservation}

One can also couple co-motions to achieve both angle monotonicity and the property that the two robots are in contact for a single connected interval. 
A review of the various cases considered in Section\ref{sec:standardproof} shows that there are just two situations, one illustrated in 
Figure~\ref{fig:NearTrivial3} (upper right) and the other in Figure~\ref{fig:Above both2} (left), where our standard-form co-motions can temporarily lose robot contact.

The optimal co-motions described for these cases can be modified to preserve contact by following the ladder motion strategy described by Icking et al.~\cite{icking} in the interval between contact configurations. The validity of the resulting straight-line co-motion in this interval is guaranteed by the fact that the interval in question either starts and end with what Icking et al. call \emph{critical} configurations (where the configuration orientation is normal to the motion trace).

For example, Figure~\ref{fig:contact-preservation}
(left) illustrates the modification required for the case illustrated in Figure~\ref{fig:Above both2}. 
Here the decoupled co-motion between critical contact configurations $(\bA[A'_0], \bB[B_0])$ and 
$(\bA[A_1], \bB[B'_0])$ is replaced by a sliding straight-line coupled motion, as described in 
Icking et al. 
Similarly, Figure~\ref{fig:contact-preservation}
(right) illustrates the modification required for the case illustrated in Figure~\ref{fig:NearTrivial3}. 
Here the decoupled co-motion between critical contact configurations $(\bA[A_0], \bB[B'_0])$ and 
$(\bA[A'_0], \bB[B_1])$ is replaced by a sliding straight-line coupled motion

\begin{figure}[ht]
\centering
\includegraphics[width=1.0\textwidth]{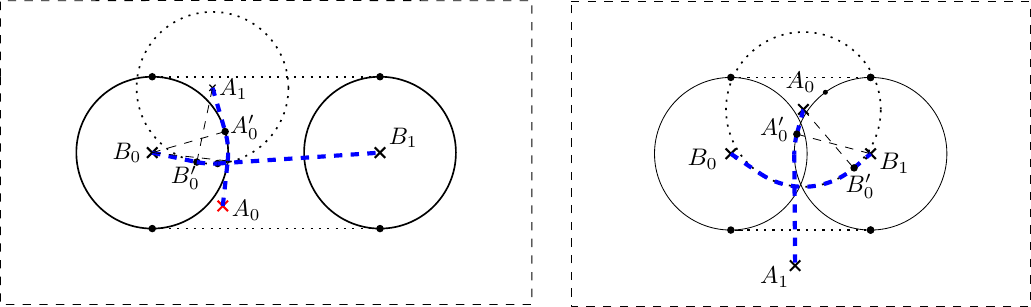}
\caption{Modification to preserve contact}
\label{fig:contact-preservation}
\end{figure}

\section{Towards a complete characterization of optimal co-motions}\label{sec:uniqueness}

Theorem~\ref{thm:standardoptimality} serves to characterize the length of optimal co-motions in terms of properties of the standard co-motions described in Section~\ref{sec:standardproof}. 
Our standard-form co-motions are not uniquely optimal. For example, Figure~\ref{fig:uniqueness} illustrates a case where there is a continuum of optimal co-motions (each tracing the boundary of an object of fixed width) whose length equals that of our standard-form co-motion (shown in blue).

\begin{figure}[ht]
\centering
\includegraphics[width=0.7\textwidth]{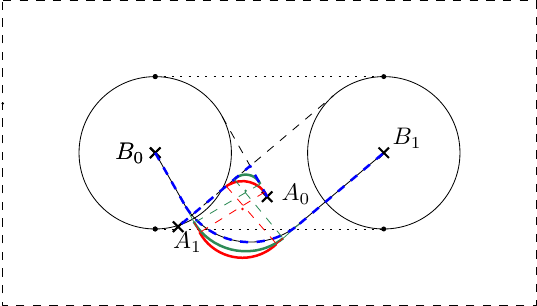}
\caption{A continuum of optimal paths}
\label{fig:uniqueness}
\end{figure}

This makes it natural to ask if there is a broader characterization of optimal co-motions. The remainder of this section demonstrates one such characterization.

 It follows from Lemma~\ref{lem:optkey} and Lemma~\ref{lem:standard} that the counter-clockwise tightness of the pair $(\widearc{\xi_{\bA}}, \widearc{\xi_{\bB}})$ is both necessary and sufficient for a convex net counter-clockwise co-motion $m = (\xi_{\bA}, \xi_{\bB})$ to be globally optimal.

Recall that $H^{\top}_{\bA\bB_0, \bA\bB_1}$ denotes the convex hull of the set of points forming the top of $(\bA\!+\!\bB)[O]$ together with the points $\Delta_{ij} = A_i - B_j$, where $i, j \in \{0, 1 \}$. 
All of the standard-form counter-clockwise optimal co-motions $(\xi_\bA, \xi_\bB)$ that we have identified have the property that the convex hull of the set $(\widearc{\xi_{\bA}} -\widearc{\xi_{\bB}})$ coincides with 
$H^{\top}_{\bA\bB_0, \bA\bB_1}$. 
It turns out that this property is characteristic of \emph{all} globally optimal collision-free net counter-clockwise co-motions.

\begin{lem}\label{lem:character}
 Let $m=(\xi_\bA, \xi_\bB)$  be any optimal collision-free co-motion from $\bA\bB_0$ to $\bA\bB_1$. 
 If $m$ is net counter-clockwise 
 then the convex hull of the set $(\widearc{\xi_{\bA}} -\widearc{\xi_{\bB}})$ coincides with 
$H^{\top}_{\bA\bB_0, \bA\bB_1}$.
\end{lem}

\begin{proof}
Note that 
$\max_{i,j \in \{0,1\}} r_{\Delta_{ij}}(\theta) =
 r_{\overline{A_0A_1} - \overline{B_0B_1}}(\theta) =
 r_{\overline{A_0A_1}}(\theta) + r_{\overline{B_0B_1}}(\pi +\theta)$.
Hence $r_{H^{\top}_{\bA\bB_0, \bA\bB_1}}(\theta)
 = \max\left(r_{\overline{A_0A_1}}(\theta) + r_{\overline{B_0B_1}}(\pi +\theta), 
r_{\bA\!+\!\bB}(\theta)
\cdot\mathds{1}_{[\theta_0, \theta_1]}(\theta) \right).$   

Suppose that the convex hull of the set $(\widearc{\xi_{\bA}} -\widearc{\xi_{\bB}})$ does not coincide with 
$H^{\top}_{\bA\bB_0, \bA\bB_1}$.
Then, for some $\theta \in S^1$, 
$r_{\widearc{\xi_{\bA}}}(\theta) + r_{\widearc{\xi_{\bB}}}(\pi + \theta) = r_{\widearc{\xi_{\bA}} -\widearc{\xi_{\bB}}}(\theta) 
\neq r_{H^{\top}_{\bA\bB_0, \bA\bB_1}}(\theta)$,
contradicting the optimality of $m$.
\end{proof}

\begin{thm}
  Let $m=(\xi_\bA, \xi_\bB)$  be any optimal collision-free co-motion from $\bA\bB_0$ to $\bA\bB_1$. 
Then the convex hull of the set $(\widearc{\xi_{\bA}} -\widearc{\xi_{\bB}})$ coincides with either $H^{\top}_{\bA\bB_0, \bA\bB_1}$ or $H^{\bot}_{\bA\bB_0, \bA\bB_1}$, whichever has a shorter boundary.
Furthermore, $\ell(m)$ is just the length of this shorter boundary less $|\overline{A_0B_0}| + |\overline{A_1B_1}|$. 
\end{thm}

\begin{proof}
We know that a globally optimal co-motion is either counter-clockwise optimal or clockwise optimal. If $m$ is counter-clockwise optimal then,
by Lemma~\ref{lem:character}, 
the convex hull of the set $(\widearc{\xi_{\bA}} -\widearc{\xi_{\bB}})$ coincides with $H^{\top}_{\bA\bB_0, \bA\bB_1}$. 
But, by Observation~\ref{obs:Cauchy2}, 
the length of the boundary of the convex hull of the set $(\widearc{\xi_{\bA}} -\widearc{\xi_{\bB}})$ is the same as $\ell(\widearc{\xi_{\bA}}) + \ell(\widearc{\xi_{\bB}})$, which, by convexity, is just $\ell(m)$.

On the other hand, if $m$ is clockwise optimal, then, by Observation~\ref{obs:flip},
$\flip(m)$ is a net counter-clockwise co-motion of the robot pair $(\flip(\bA), \flip(\bB))$ from configuration 
$\flip(\bA\bB_0)$ to $\flip(\bA\bB_1)$. But since 
$\ell(\flip(m))= \ell(m)$ and 
$H^{\bot}_{\bA\bB_0, \bA\bB_1} = 
\flip(H^{\top}_{\flip(\bA\bB_0),\flip( \bA\bB_1)})$,
the result follows.
\end{proof}

\begin{obs}
Note that in all but one case, discussed in subsection~\ref{subsec:rightofboth} and illustrated in Figure~\ref{fig:A1right}, the trace of the motion $\xi_{\bA -\bB}$, defined as $(\xi_{\bA -\bB}(t) = \xi_\bA(t) - \xi_\bB(t)$
follows the boundary of the convex hull of the set $(\widearc{\xi_{\bA}} -\widearc{\xi_{\bB}})$.  
\end{obs}

\section{Conclusions}
\label{sec:con}

We have provided constructions of shortest collision-avoiding motions for two arbitrary centrally-symmetric convex robots in a planar obstacle-free environment, taking an arbitrary initial configuration of the robots into an arbitrary goal configuration. 
The total path length of the coordinated motions is neatly characterized by a simple integral.
The coordinated motion has the property that it is \emph{decoupled}, so that only one robot is moving at any given time. 
However, it can also be \emph{coupled} so that the relative orientation of the robots changes monotonically over time. 
The coupled co-motion has the additional property that the robots are in contact for a connected interval of time, that is, once they move out of contact, they are never in contact again. 

One key ingredient that allowed for the integral characterization of the optimal motion was Cauchy's surface area formula, which in 2D allows one to transform lengths of convex traces to an integral in $S^1$. 
As far as we know, the tools that we have used are limited to the case when the robots are in 2D. Indeed, when the robots are spheres in 3D, even if the initial and goal positions of the robot reside in a common plane, we have not been able to show that the shortest path stays within this plane (except in special cases). 

Another natural extension of our work would be to remove the dependence on central-symmetry. Cauchy's formula applies to any convex trace. In various cases of our proof, however, we rely on central-symmetry to reflect the paths of the robots. The arguments of Section~\ref{sec:special} are one such example of this. Outside of these problematic cases, many simple cases of optimal co-motions for general convex robots can still be proven with the framework presented in this paper.
The extension of our results to 3D, even for spherical robots, as well as the extension to convex non-centrally-symmetric robots, constitute enticing subjects for future exploration.

\bibliographystyle{plain}
\bibliography{motionplanning}

\end{document}